\documentclass[sn-mathphys-num]{sn-jnl}

\usepackage{graphicx}%
\usepackage{multirow}%
\usepackage{amsmath,amssymb,amsfonts}%
\usepackage{amsthm}%
\usepackage{mathrsfs}%
\usepackage[title]{appendix}%
\usepackage{xcolor}%
\usepackage{textcomp}%
\usepackage{manyfoot}%
\usepackage{booktabs}%
\usepackage{algorithm}%
\usepackage{algorithmicx}%
\usepackage{algpseudocode}%
\usepackage{listings}%
\usepackage{subfigure}
\usepackage{enumerate}

\theoremstyle{thmstyleone}%
\newtheorem{theorem}{Theorem}[section]
\newtheorem{lemma}[theorem]{Lemma}%

\theoremstyle{thmstyletwo}%
\newtheorem{remark}{Remark}%

\theoremstyle{thmstylethree}%
\newtheorem{definition}{Definition}[section]%

\begin{document}

\title[Compressed Newton Thresholding]{Compressed Newton-direction-based Thresholding Methods for Sparse Optimization Problems\footnote{This work is supported by the Natural Science Foundation of China under the grants 12301400, 12471295 and 12426306, Guangdong Basic and Applied Basic Research Foundation
(2024A1515011566), and Hetao Shenzhen-Hong Kong Science and Technology Innovation Cooperation
Zone Project (HZQSWS-KCCYB-2024016).}
}


\author[1]{\fnm{Nan} \sur{Meng}}\email{Nan.Meng@nottingham.edu.cn}

\author*[2]{\fnm{Yun-Bin} \sur{Zhao}}\email{Yunbinzhao@cuhk.edu.cn}

\affil[1]{\orgdiv{School of Mathematical Sciences}, \orgname{University of Nottingham Ningbo China}, \orgaddress{\city{Ningbo}, \state{Zhejiang Province} \postcode{315100}, \country{China}}}

\affil[2]{\orgdiv{Shenzhen International Center for Industrial and Applied Mathematics, SRIBD}, \orgname{The Chinese University of Hong Kong}, \orgaddress{\city{Shenzhen}, \state{Guangdong Province} \postcode{518100}, \country{China}}}

\abstract{
Thresholding algorithms for sparse optimization problems involve two key components: search directions and thresholding strategies.
In this paper, we use the compressed Newton direction as a search direction, derived by confining the classical Newton step to a low-dimensional subspace and embedding it back into the full space with diagonal regularization.
This approach significantly reduces the computational cost for finding the search direction while maintaining the efficiency of Newton-like methods.
Based on this new search direction, we propose two major classes of algorithms by adopting hard or optimal thresholding: the compressed Newton-direction-based thresholding pursuit (CNHTP) and compressed Newton-direction-based optimal thresholding pursuit (CNOTP).
We establish the global convergence of the proposed algorithms under the restricted isometry property.
Experimental results demonstrate that the proposed algorithms perform comparably to several state-of-the-art methods in terms of success frequency and solution accuracy for solving the sparse optimization problem.
}


\maketitle

\section{Introduction}
The sparse optimization problem is essential to many practical applications, such as compressed sensing \cite{foucart2013mathematical}, signal processing \cite{elad2010sparse}, machine learning \cite{sra2011optimization}, feature selection \cite{gui2016feature}, and model overfitting prevention \cite{tibshirani1996regression}.
A widely adopted sparse optimization model is the cardinality-constrained least-squares problem:
\begin{equation} \label{prob1}
\min_{x \in \mathbb{R}^n} \left\{\|y-A x\|_2^2:\|x\|_0 \leq k \right\},
\end{equation}
where $A \in \mathbb{R}^{m \times n}$ is a given matrix $(m \ll n)$, $y$ is a given vector, the $\ell_0$-norm $\|\cdot\|_0$ counts the number of nonzero entries of a vector, and $k$ is a given integer.
$x \in \mathbb{R}^n$ satisfying the constraint in \eqref{prob1} is called a $k$-sparse vector.
A number of variants and relaxations of \eqref{prob1} have been extensively studied; see, for instance, \cite{eldar2012compressed}, \cite{foucart2013mathematical} and \cite{zhao2018sparse}.

There are several classes of numerical methods for solving sparse optimization problems, including the convex optimization methods such as LASSO \cite{tibshirani1996regression} and basis pursuit \cite{chen2001atomic, donoho2001uncertainty, donoho2006stability}, greedy pursuit methods such as  \cite{pati1993orthogonal, tropp2004greed, dai2009subspace, donoho2012sparse, wang2015recovery, zhao2023dynamic}, thresholding methods such as \cite{tibshirani1996regression, beck2009fast, blumensath2008iterative, blumensath2009iterative, foucart2011hard, blumensath2010normalized, tanner2013normalized, zhao2020optimal}, and Bayesian sparse-learning methods \cite{ji2008bayesian}.
Among these methods, the thresholding method is one of the most widely used approaches due to its simple iterative scheme.
It can be viewed as a combination of the traditional descent optimization methods and a certain thresholding technique used to generate sparse iterates.
As described in \cite{meng2022partial}, a general iterative scheme for thresholding methods takes the form
\begin{equation} \label{model1}
x^{(p+1)}=\mathcal{T}_k (x^{(p)}+\lambda \mathbf{d}),
\end{equation}
where $\mathcal{T}_k(\cdot)$ represents a thresholding operator that returns a $k$-sparse vector, $\lambda \in \mathbb{R}$ denotes the stepsize, and $\mathbf{d} \in \mathbb{R}^{n}$ is a search direction at the current iterate $x^{(p)} \in \mathbb{R}^{n}$.

The steepest descent direction and the Newton direction are two popular search directions \cite{nocedal2006numerical}.
The method based on the Newton direction may achieve a fast local convergence and thus can generally solve the problem faster than the first-order methods.
However, a standard Newton step is inapplicable to underdetermined problems such as \eqref{prob1}, in which the Hessian of the objective function is singular.
To address this issue, various Newton-type modifications have been proposed \cite{dembo1982truncated, martens2010deep}, and have also been generalized to solve sparse optimization problems \cite{yuan2014newton, yuan2017newton, chen2017fast, meng2020newton, zhou2021newton, zhou2021global}.

As indicated in \eqref{model1}, the thresholding operator $\mathcal{T}_k$ plays an essential role in generating the sparse iterate.
Different choices of $\mathcal{T}_k$ lead to different iterative methods.
Many existing thresholding methods fall within the framework of the iterative scheme \eqref{model1}, including soft thresholding \cite{donoho1995denoising, tibshirani1996regression, beck2009fast}, hard thresholding \cite{blumensath2008iterative, foucart2011hard, blumensath2010normalized, tanner2013normalized} and optimal thresholding \cite{zhao2020optimal, zhao2021analysis}.
In this paper, we focus on the methods using hard thresholding and/or optimal thresholding operators.
Hard thresholding operator retains the $k$ largest entries in magnitude of a vector and sets the remaining entries to zero.
Iterative methods using hard thresholding and steepest descent direction were studied in such references as \cite{blumensath2008iterative, blumensath2010normalized, foucart2011hard, tanner2013normalized}, and Newton-type directions in \cite{jing2014quasi, wang2018new, meng2020newton}.
Optimal thresholding (OT), distinct from hard thresholding, selects the best $k$ terms of a vector that minimize the objective function among all possible $k$ terms of the vector.
It has been combined with the steepest descent direction \cite{zhao2020optimal, zhao2021analysis}, Newton-type directions \cite{meng2022newton}, and certain memorable search directions \cite{sun2023heavy} to develop specific approaches for sparse optimization problems.

At present, the search direction is typically either the full gradient of the objective function or its Newton counterpart.
In this paper, we study the following iterative scheme:
\begin{equation} \label{model3}
x^{(p+1)}=\mathcal{T}_k (x^{(p)}+\lambda \mathbf{d}_{\text{compress}}),
\end{equation}
where the search direction $\mathbf{d}_{\text{compress}}$ itself is compressible.
Indeed, the sparse search direction can be viewed as a special case within this framework.
The framework (\ref{model3}) extends the approach in \cite{meng2022partial} where the search direction was taken as the partial steepest descent direction and $\mathcal{T}_k$ was taken as OT.
We introduce the so-called compressed Newton direction which is a novel compressed search direction proposed here for the first time.

In this paper, we first prove that the compressed Newton direction is a genuine descent direction for the objective function.
Then we incorporate it into the iterative framework \eqref{model3} to obtain two types of algorithms by using hard and optimal thresholding operators.
The first type is referred to as the compressed Newton-direction-based hard thresholding (CNHT) method, and its enhanced counterpart is called the compressed Newton-direction-based hard thresholding pursuit (CNHTP).
The second type is called the compressed Newton-direction-based optimal thresholding (CNOT) together with its enhanced counterpart called the compressed Newton-direction-based optimal thresholding pursuit (CNOTP).
Under the restricted isometry property (RIP), we establish the global convergence results for these algorithms.
The main theoretical results are summarized in Theorems \ref{thm-cnt}, \ref{thm-cntp}, \ref{thm-cnot}, and \ref{thm-cnotp}, respectively.
Moreover, numerical experiments are conducted on random sparse optimization problems to demonstrate the performance of our algorithms and to compare with several existing methods.

This work is organized as follows.
Section \ref{intro} introduces the main algorithmic framework. 
Section \ref{theo} conducts the theoretical analysis for the proposed algorithms.
Section \ref{numerical} demonstrates the numerical results of the proposed methods.

\textbf{Notation:} The $\ell_0$-norm, $\| \cdot \|_0$, counts the number of nonzero entries of a vector.
The $\ell_2$-norm $\|x\|_2$ is defined as $\textstyle \sqrt{\sum_{i=1}^n x_i^2}$, where $x \in \mathbb{R}^n$.
Let $\Omega \subseteq[N] := \{1, \ldots, n\}$ denote an index set with cardinality $|\Omega| \leq t$, where $t$ is a certain integer number.
The complement of $\Omega$ is denoted by $\overline{\Omega} := [N] \backslash \Omega$.
For a matrix $A \in \mathbb{R}^{m \times n}$, $A_{\Omega} \in \mathbb{R}^{m \times t}$ denotes the submatrix formed by selecting the columns of $A$ indexed by $\Omega$, and $A_{\Omega |} \in \mathbb{R}^{m \times n}$ is the matrix obtained from $A$ by zeroing out all columns of $A$ outside $\Omega$.
For a square matrix $B \in \mathbb{R}^{n \times n}, B_{\Omega \times \Omega} \in \mathbb{R}^{t \times t}$ denotes the principal submatrix of $B$ formed by the rows and columns indexed by $\Omega$.
Let $I$ denote the identity matrix and $\mathbf{e}$ the vector of ones, whose dimensions are clear from the content.
The operator $\mathcal{T}_k(\cdot)$ produces a $k$-sparse vector.
The hard thresholding operator $\mathcal{H}_k(\cdot)$ retains the $k$ largest-magnitude entries of a vector and sets the rest to zero.
Supp$(x) = \{i : x_i \neq 0\}$ is called the support of $x$, and the support of $\mathcal{H}_k(\cdot)$ is denoted by $\mathcal{L}_k(\cdot)$.
The optimal $k$-thresholding operator $\mathcal{Z}_k^{\#}(\cdot)$ selects the $k$ entries of a vector that minimize a given objective function.
The Hadamard product $u \otimes v$ denotes the elementwise product of two vectors $u$ and $v$.
$\lceil \cdot \rceil$ denotes the ceiling function, which maps a real number to the smallest integer greater than or equal to it.

\section{Algorithms} \label{intro}
In this section, we describe the compressed Newton search direction and the algorithms based on it.

\subsection{Compressed Newton search direction}
The approach in \cite{meng2022partial} employs the partial negative gradient as a search direction, i.e.,
\[
\mathbf{d}_{\text{sparse}} = {\cal H}_q (A^{\top}(y-Ax^{(p)})).
\]
This search direction is used to update the current iterate within a subspace defined by the index set
\[
\Omega^{p} := \mathcal{L}_q (A^{\top}(y-A x^{(p)})).
\]
Let $A_{\Omega^p} \in \mathbb{R}^{m \times q}$ and $A_{\overline{\Omega^p}} \in \mathbb{R}^{m \times(n-q)}$ denote the submatrices of $A$ consisting of the columns indexed by $\Omega^p$ and its complement $\overline{\Omega^p}$, respectively.
Recall that $A_{\Omega^p |}$, $A_{\overline{\Omega^p} |} \in \mathbb{R}^{m \times n}$ are the matrices obtained by zeroing out all columns of $A$ not indexed by $\Omega^p$ and $\overline{\Omega^p}$, respectively.
Given $\Omega^p$, the objective function in the problem \eqref{prob1} can be rewritten as
\begin{align}
y-Ax
&= y-(A_{\Omega^{p} |} + A_{\overline{\Omega^{p}}|}) x \nonumber \\
&= \tilde{y} - A_{\Omega^{p} |} x \nonumber \\
&= \tilde{y} - A_{\Omega^{p}} x_{\Omega^{p}}, \label{prob2}
\end{align}
where $x_{\Omega^{p}}$ is the restriction of $x$ to $\Omega^{p}$, and  $\tilde{y} = y - A_{\overline{\Omega^{p}}|} x$.

Define $g(x_{\Omega^{p}}) = \left\|\tilde{y} - A_{\Omega^{p}} x_{\Omega^{p}} \right\|_2^2$.
The expressions of gradient and Hessian of $g(x_{\Omega^{p}})$ are given respectively by
\[
\begin{aligned}
\nabla g\left(x_{\Omega^{p}}\right) = -A_{\Omega^{p}}^{\top} \left(\tilde{y} - A_{\Omega^{p}} x_{\Omega^{p}}\right)~~\text{and}~~
\nabla^2 g\left(x_{\Omega^{p}}\right) = A_{\Omega^{p}}^{\top} A_{\Omega^{p}}.
\end{aligned}
\]
When $q = |\Omega^{p}| \ll m$, the Hessian is typically invertible under suitable assumptions, such as the RIP (see Definition \ref{def-ric} later).
The standard Newton search direction for minimizing $g(x_{\Omega^{p}})$ is
\begin{align}
\mathbf{d}_{g\left(x_{\Omega^{p}}\right)\text{-ND}} 
&= - \left(\nabla^2 g\left(x_{\Omega^{p}}\right)\right)^{-1} \nabla g\left(x_{\Omega^{p}}\right)  \nonumber \\
&= (A_{\Omega^{p}}^{\top} A_{\Omega^{p}})^{-1} A_{\Omega^{p}}^{\top} \left(\tilde{y}-A_{\Omega^{p}} x_{\Omega^{p}}\right) \nonumber\\
&= \left(A_{\Omega^{p}}^{\top} A_{\Omega^{p}}\right)^{-1} A_{\Omega^{p}}^{\top} \left(y-Ax\right), \label{666D}
\end{align} 
where the last equality follows from the identity \eqref{prob2}.
To embed this subspace vector into the full space, we define the matrix $M(\Omega^p) \in \mathbb{R}^{n \times n}$ as follows: 
\begin{equation} \label{hessian}
M(\Omega^p)_{\Omega^p \times \Omega^p}=(A_{\Omega^p}^{\top} A_{\Omega^p})^{-1}
\end{equation}
and other entries $\left[M\left(\Omega^p\right)\right]_{i j} \equiv 0$ for $i \notin \Omega^p$ or $j \notin \Omega^p$ and $i \neq j$, and diagonal entries  $\left[M\left(\Omega^p\right)\right]_{i i}=\alpha$ for $i \notin \Omega^p$. 
Motivated by (\ref{666D}), in this paper, we define the compressed Newton search direction as
\begin{equation} \label{direction}
\mathbf{d}_{\text{CN}} := M(\Omega^p) Z(\Omega^p) A^{\top}(y-A x^{(p)}),
\end{equation}
where $Z(\Omega^{(p)})=\operatorname{diag}\left(z_1, z_2, \ldots, z_n\right)$ is a diagonal matrix with entries
\begin{equation} \label{diamatrix}
z_i= \begin{cases}1, & \text { if } i \in \Omega^{(p)}, \\ \gamma, & \text { if } i \notin \Omega^{(p)},\end{cases}
\end{equation}
and $\gamma > 0$ is a parameter.
Clearly, \eqref{666D} can be seen as the extreme case of (\ref{direction}) corresponding to the case $ \gamma=0.$ 
In \eqref{direction}, we compute a standard Newton direction within the subspace defined by $\Omega^p$, and then we inflate it to the full space by assigning $\alpha > 0$ to the diagonal entries outside $\Omega^p$.
Moreover, the matrix $Z(\Omega^p)$ emphasizes the best $q$-terms of the gradient while attenuating the remaining entires via the scaling factor $\gamma > 0$, which ensures that the search direction containing sufficient information.
When $q \ll n$, this approach significantly reduces the computational cost of finding a Newton-type search direction in $\mathbb{R}^n$.

It is not difficult to check that $\mathbf{d}_{\mathrm{CN}}$ is a descent direction for the objective $f(x)=\|y-A x\|_2^2$ at $x=x^{(p)}$. Indeed, one has
\[
\begin{aligned}
\left(\nabla f(x)\right)^{\top} \mathbf{d}_{\text{CN}} 
=& -(A^{\top}(y-A x))^{\top} ~ M(\Omega^p) Z(\Omega^p) (A^{\top}(y-A x))_{\Omega^p} \\
=& -(A^{\top}(y-A x))_{\Omega^p}^{\top} ~ M(\Omega^p) (A^{\top}(y-A x))_{\Omega^p} - \gamma \alpha \left\|(A^{\top}(y-A x))_{\overline{\Omega^p}}\right\|_2^2 \\
=& -(A^{\top}(y-A x))_{\Omega^p}^{\top} (A_{\Omega^p}^{\top} A_{\Omega^p})^{-1} (A^{\top}(y-A x))_{\Omega^p} \\
& - \gamma \alpha \left\|(A^{\top}(y-A x))_{\overline{\Omega^p}}\right\|_2^2 \\
<& 0,
\end{aligned}
\]
provided $A^{\top}(y-A x) \neq \mathbf{0}$.

\subsection{Compressed Newton-direction-based algorithms}
Using $\mathbf{d}_{\text{CN}}$ in \eqref{direction} leads to the iterative scheme
\begin{equation} \label{model2}
x^{(p+1)} = \mathcal{T}_k (u^{(p)}) = \mathcal{T}_k(x^{(p)}+\lambda \mathbf{d}_{\text{CN}}),
\end{equation}
where $u^{(p)} = x^{(p)} + \lambda \mathbf{d}_{\text{CN}}$ with stepsize $\lambda$.
Choosing $\mathcal{T}_k$ as the hard thresholding operator in \eqref{model2} leads to the compressed Newton-direction-based thresholding (CNHT) algorithm (see Algorithm \ref{CNHT}), i.e.,
\[
x^{(p+1)} = {\cal H}_k (u^{(p)}).
\]
It can be further enhanced by a post-projection (or pursuit) step, which aims to reduce the objective function over the selected support, e.g., see \eqref{pursuit} for details.
This enhanced version of CNHT, termed the \textit{compressed Newton-direction-based thresholding pursuit} (CNHTP) algorithm is detailed in Algorithm \ref{CNTP}.
Hard thresholding appears twice in this framework, but serves two distinct purposes.
The first application, in \eqref{alg1-CN}, acts on the negative gradient and retains the largest entries in magnitude to select the subspace for computing the Newton direction.
The second step is to generate a $k$-sparse vector on the updated iteration based on the search direction to ensure the feasibility of iterate to \eqref{prob1}.

\begin{algorithm}
\caption{Compressed Newton-direction-based thresholding (CNHT)} \label{CNHT}
\begin{enumerate}[Step 1.]
\item (Initialization) Input matrix $A$, vector $y$, sparsity level $k$, integer number $q \in [k ,2k]$, stepsize $\lambda$, and initial point  $x^0 = \mathbf{0}$.
\item (Compressed Newton direction) Compute the intermediate vector
\begin{equation} \label{alg1-CN}
u^{(p)} =x^{(p)}+\lambda M(\Omega^p) Z(\Omega^p) A^{\top} (y-A x^{(p)} ),
\end{equation}
where $M(\Omega^p)$ and $Z(\Omega^p)$ are given in \eqref{hessian} and \eqref{diamatrix}, respectively.
\item Let
\[
x^{(p+1)} = {\cal H}_k (u^{(p)}).
\]
\end{enumerate}
Repeat Steps 2 and 3 until a stopping criterion is met.
\end{algorithm}

\begin{algorithm}
\caption{Compressed Newton-direction-based thresholding pursuit (CNHTP)} \label{CNTP}
\begin{enumerate}[Step 1.]
\item Same as Algorithm \ref{CNHT}.
\item Same as Algorithm \ref{CNHT}.
\item Let
\begin{align}
S^{(p+1)} &= \operatorname{supp} ({\cal H}_k (u^{(p)})) \nonumber \\
x^{(p+1)} & = \arg \min _z \left\{ \|y-A z\|_2^2: \operatorname{supp}(z) \subseteq S^{(p+1)} \right\}. \label{pursuit}
\end{align}
\end{enumerate}
Repeat Steps 2 and 3 until a stopping criterion is met.
\end{algorithm}

Replacing ${\cal H}_k (\cdot)$ in \eqref{model2} by the optimal thresholding operator $\mathcal{Z}^{\#}_k$ leads to the iterative scheme
\[
x^{(p+1)} = \mathcal{Z}^{\#}_k (u^{(p)}) = u^{(p)} \otimes \hat{w}^{(p)},
\]
where $\hat{w}^{(p)}$ solves the binary optimization problem
\begin{equation} \label{ot}
\min_{\hat{w}} \left\{\|y-A(u^{(p)} \otimes \hat{w})\|_2^2: \mathbf{e}^{\top} \hat{w}=k, \hat{w} \in\{0,1\}^n\right\},
\end{equation}
and $\mathbf{e}$ denotes the vector of ones.
However, as indicated in \cite{zhao2020optimal}, \eqref{ot} is a binary optimization problem which might be expensive to solve in general, and thus one can relax the binary constraint in \eqref{ot} to obtain its convex relaxation
\[
\min_{w} \left\{\|y-A(u^{(p)} \otimes w)\|_2^2: \mathbf{e}^{\top} w=k, \mathbf{0} \le w \le \mathbf{e} \right\},
\]
to which the solution $w^{(p)}$ might not be exactly $k$-sparse.
Therefore, a hard thresholding operator is applied to the vector $u^{(p)} \otimes w^{(p)}$ to generate a $k$-sparse iterate.
This relaxation leads to another two algorithms: compressed Newton direction with optimal thresholding (CNOT) and its enhanced extension called CNOTP (see Algorithms \ref{CNOT} and \ref{CNOTP}, respectively).

\begin{algorithm}
\caption{Compressed Newton-direction-based optimal thresholding (CNOT)} \label{CNOT}
\begin{enumerate}[Step 1.]
\item Same as Algorithm \ref{CNHT}.
\item Same as Algorithm \ref{CNHT}.
\item Let
\[
\begin{aligned}
w^{(p)} &=\arg \min _{w} \{\|y-A(w \otimes u^{(p)})\|_{2}^{2}: \mathbf{e}^{\top} w=k, \mathbf{0} \le w \le \mathbf{e} \}, \\
x^{(p+1)} &= {\cal H}_k (w^{(p)} \otimes u^{(p)}).
\end{aligned}
\]
\end{enumerate}
Repeat Steps 2 and 3 until a stopping criterion is met.
\end{algorithm}

\begin{algorithm}
\caption{Compressed Newton-direction-based optimal thresholding pursuit (CNOTP)} \label{CNOTP}
\begin{enumerate}[Step 1.]
\item Same as Algorithm \ref{CNHT}.
\item Same as Algorithm \ref{CNHT}.
\item Let
\begin{align}
w^{(p)} &=\arg \min _{w} \{\|y-A(w \otimes u^{(p)})\|_{2}^{2}: \mathbf{e}^{\top} w=k, \mathbf{0} \le w \le \mathbf{e} \}, \nonumber \\
S^{(p+1)} &=  \operatorname{supp} ({\cal H}_k (w^{(p)} \otimes u^{(p)})), \nonumber \\
x^{(p+1)} &= \arg \min_z \left\{\|y-A z\|_2^2: \operatorname{supp}(z) \subseteq S^{(p+1)} \right\}. \nonumber
\end{align}
\end{enumerate}
Repeat Steps 2 and 3 until a stopping criterion is met.
\end{algorithm}

\begin{remark} 
The algorithms in this study differ from existing Newton-type methods in the sparse optimization literature primarily in their choice of search direction.
The approaches in \cite{zhou2021global, zhou2022gradient} employ gradient-based directions followed by a Newton correction after thresholding.
The method in \cite{meng2020newton} computes a Newton-like step within $\mathbb{R}^n$ to avoid Hessian singularity.
Other methods, such as those in \cite{zhou2021newton, liao2024subspace, ye2025subspace}, adaptively switch between Newton and proximal gradient steps based on certain feasibility conditions.
In contrast, the search direction in this study is a compressed Newton one in $\mathbb{R}^q$ ($q \ll n$), which incorporates scaled components of the rest of gradient.
\end{remark}

\section{Theoretical analysis} \label{theo}
We establish the convergence of the proposed algorithms under a matrix property in terms of the restricted isometry constant (RIC), defined as follows.

\begin{definition} \emph {\cite{candes2005decoding}} \label{def-ric}
The restricted isometry constant of order $k$ for a matrix $A \in \mathbb{R}^{m \times n}$ ($m \ll n$), denoted as $\delta_k \in [0,1)$, is defined as the smallest $0 \le \delta < 1$ such that
\[
\begin{aligned}
(1-\delta)\|x\|_2^2 \leq\|A x\|_2^2 \leq(1+\delta)\|x\|_2^2
\end{aligned}
\]
for all $k$-sparse vectors $x$.
The matrix $A$ is said to satisfy the restricted isometry property (RIP) if $\delta_k < 1$.
\end{definition}

We now present three technical lemmas that are useful for the convergence analysis of the proposed algorithms.
The first lemma characterizes a property of matrices that satisfy the restricted isometry property, and the second concerns a property of the hard thresholding operator.

\begin{lemma} \emph{\cite{foucart2013mathematical}} \label{lemma-rip}
Let $A \in \mathbb{R}^{m \times n}$ satisfy the RIP of order $t$ with constant $\delta_t$.
For any vector $v \in \mathbb{R}^n$ and any index set $\Lambda \subseteq [N]$, one has
\begin{enumerate}[(i)]
\item $\left\|\left(\left(I-A^T A\right) v\right)_{\Lambda}\right\|_2 \le \delta_t\|v\|_2$, if $|\Lambda \cup \operatorname{supp}(v)| \leq t$;
\item $\left\|\left(A^T v\right)_{\Lambda}\right\|_2 \leq \sqrt{1+\delta_t}\|v\|_2$, if $|\Lambda| \leq t$.
\end{enumerate}
\end{lemma}

\begin{lemma} \emph{\cite{zhao2023improved} } \label{lemma-hard}
For any vector $z \in \mathbb{R}^n$ and any $k$-sparse vector $u \in \mathbb{R}^n$ (i.e., $\|u\|_0 \leq k$), one has
\[
\left\|u-\mathcal{H}_k(z)\right\|_2 \leq \varphi \left\|(u-z)_{\Lambda \cup S^*}\right\|_2,
\]
where $\varphi := \frac{\sqrt{5}+1}{2}$, $\Lambda=\operatorname{supp}(u)$ and $S^*=\operatorname{supp}\left(\mathcal{H}_k(z)\right)$.
\end{lemma}

The next lemma is crucial to our analysis for the proposed algorithms.

\begin{lemma} \label{lemma1}
Let $A \in \mathbb{R}^{m \times n}$ be a matrix satisfying the RIP, and $\lambda>0$ be a constant.
For any index set $\Omega \subseteq [N]$ with cardinality $|\Omega| \le q$, matrix $M(\Omega) \in \mathbb{R}^{n \times n}$ is given by
\[
M(\Omega)_{\Omega \times \Omega} = (A_{\Omega}^{\top} A_{\Omega})^{-1},
\]
with other off-diagonal entries equal to zero, and diagonal entries $[M(\Omega)]_{ii} = \alpha$ for $i \notin \Omega$.
Let $Z(\Omega) \in \mathbb{R}^{n \times n}$ be a diagonal matrix given as
\[
Z(\Omega) = \operatorname{diag}(z_1, z_2, \ldots, z_n),
\]
where  $z_i=1$ for $i \in \Omega$ and $z_i=\gamma > 0$ for $i \notin \Omega$.
Let the parameters $\alpha$ and $\gamma$ satisfy $\frac{1}{1+\delta_q} \leq \alpha \gamma \leq \frac{1}{1-\delta_q}$.
Then,
\begin{enumerate}[(i)]
\item For any $2k$-sparse vector $v \in \mathbb{R}^n$ and a set $\Lambda \subseteq [N]$ with cardinality $|\Lambda| \le k$, one has
\begin{equation} \label{lemma1-4}
\|[I - \lambda M(\Omega) Z(\Omega) A^{\top} A v]_{\Lambda}\|_2
\le 
\left(\delta_{3k} + \Delta(\lambda, \delta_q)\left(1 + \sqrt{\left\lceil \tfrac{n}{k} \right\rceil}\, \delta_{3k}\right)\right) 
\|v\|_2,
\end{equation}
where
\begin{equation} \label{lemma1-10}
\Delta (\lambda, \delta_q) := \max \left\{\left|1-\tfrac{\lambda}{1-\delta_q}\right|,\left|1-\tfrac{\lambda}{1+\delta_q}\right|\right\} .
\end{equation}
\item For any vector $u \in \mathbb{R}^n$, one has
\begin{equation} \label{lemma1-ii}
\|M(\Omega) Z(\Omega) A^{\top} u\|_2 \le \frac{\sqrt{\left\lceil\frac{n}{k}\right\rceil\left(1+\delta_k\right)}}{1-\delta_q} \|u\|_2.
\end{equation}
\end{enumerate}
\end{lemma}

\begin{proof}
We first show case (i).
Note that
\[
\begin{aligned}
I - \lambda M(\Omega) Z(\Omega) A^{\top} A
&= (I - A^{\top} A) + (I - \lambda M(\Omega) Z(\Omega))A^{\top} A,
\end{aligned}
\]
by which we obtain that
\begin{align}
\|[I - \lambda M(\Omega) Z(\Omega) A^{\top} A v]_{\Lambda}\|_2 
&= \| [ (I - A^{\top} A)v ]_{\Lambda} + [ (I - \lambda M(\Omega) Z(\Omega)) A^{\top} A v ]_{\Lambda} \|_2 \nonumber \\
&\le \|[ (I - A^{\top} A) v ]_{\Lambda}\|_2
+ \|[ (I-\lambda M(\Omega) Z(\Omega)) A^{\top} A v ]_{\Lambda}\|_2 \nonumber \\
&\le \delta_{3k} \|v\|_2 
+ \|[ (I-\lambda M(\Omega) Z(\Omega)) A^{\top} A v ]_{\Lambda}\|_2, \label{lemma1-5}
\end{align}
where inequality \eqref{lemma1-5} follows from Lemma \ref{lemma-rip} (i), since $|\operatorname{supp}(v) \cup \Lambda| \leq 3 k$.
We now bound the second term in \eqref{lemma1-5}. Observe that
\begin{align}
\|[ (I-\lambda M(\Omega) Z(\Omega)) A^{\top} A v ]_{\Lambda}\|_2 
&\le \|(I-\lambda M(\Omega) Z(\Omega)) A^{\top} A v\|_2 \nonumber \\
&\le \|I-\lambda M(\Omega) Z(\Omega)\|_2 \|A^{\top} A v\|_2 \nonumber \\
&= \|I-\lambda M(\Omega) Z(\Omega)\|_2 \|(I+A^{\top} A-I) v\|_2 \nonumber \\
&\le \|I-\lambda M(\Omega) Z(\Omega)\|_2 \left(\left\|v\right\|_2 + \|(A^{\top} A-I) v\|_2 \right). \label{lemma1-3}
\end{align}
Next, we bound the term $\|I - \lambda M(\Omega) Z(\Omega)\|_2$.
Let $P \in \mathbb{R}^{n \times n}$ be a permutation matrix such that
\begin{equation} \label{lemma1-7}
\widetilde{M}(\Omega):=P^{\top} M(\Omega) P=\left[\begin{array}{cc}
\left(A_{\Omega}^{\top} A_{\Omega}\right)^{-1} & 0 \\
0 & \alpha I
\end{array}\right], \quad \widetilde{Z}(\Omega):=P^{\top} Z(\Omega) P
=\left[\begin{array}{cc}
I & 0 \\
0 & \gamma I
\end{array}\right].
\end{equation}
It follows that $M(\Omega) Z(\Omega)=P \widetilde{M}(\Omega) \widetilde{Z}(\Omega) P^{\top}$.
Since $P$ is orthogonal, we have
\[
\|I - \lambda M(\Omega) Z(\Omega)\|_2 = \| P (I - \lambda \widetilde{M}(\Omega) \widetilde{Z}(\Omega)) P^{\top} \|_2= \|I - \lambda \widetilde{M}(\Omega) \widetilde{Z}(\Omega)\|_2,
\]
where
\[
I - \lambda \widetilde{M}(\Omega) \widetilde{Z}(\Omega) =
\begin{bmatrix}
I - \lambda (A_{\Omega}^{\top} A_{\Omega})^{-1} & 0 \\
0 & I - \lambda \alpha \gamma I
\end{bmatrix}.
\]
The RIP inequality implies that the eigenvalues of $(A_{\Omega}^{\top} A_{\Omega})^{-1}$ lie in $\left[\frac{1}{1+\delta_q}, \frac{1}{1-\delta_q}\right]$.
Therefore, we have
\[
\|I - \lambda \widetilde{M}(\Omega) \widetilde{Z}(\Omega)\|_2
= \max \left\{ \left|1 - \tfrac{\lambda}{1 - \delta_q}\right|,\; \left|1 - \tfrac{\lambda}{1 + \delta_q}\right|,\; |1 - \lambda \alpha \gamma| \right\}.
\]
Define $$\Delta(\lambda, \delta_q)
:= \max \left\{\left|1-\frac{\lambda}{1-\delta_q}\right|,\left|1-\frac{\lambda}{1+\delta_q}\right|\right\} .$$
It is easy to verify that when $\frac{1}{1+\delta_q} \le \alpha \gamma \le \frac{1}{1-\delta_q}$, one has  $|1 - \lambda \alpha \gamma| \le \Delta(\lambda, \delta_q)$.From the above analysis, we conclude that
\begin{equation} \label{lemma1-1}
\|I-\lambda M(\Omega) Z(\Omega)\|_2 = \|I - \lambda \widetilde{M}(\Omega) \widetilde{Z}(\Omega)\|_2= \Delta(\lambda, \delta_q).
\end{equation}

Next, partition the index set $[N] = \{1,2,\ldots,n\}$ into $t := \left\lceil \tfrac{n}{k} \right\rceil$ disjoint subsets $N_1, \ldots, N_t$, each containing at most $k$ elements, so that $N_i \cap N_j = \emptyset$ for $i \neq j$.  
Then we write
\begin{equation} \label{lemma1-6}
\|(A^{\top} A - I) v\|_2 
= \|[(A^{\top} A - I) v]_{N_1 \cup N_2 \cup \dots \cup N_t}\|_2.
\end{equation}
Applying Lemma \ref{lemma-rip} with $|\operatorname{supp}(v) \cup N_i | \leq 3 k$, we obtain
\[
\|[(A^{\top} A - I)v]_{N_1 \cup N_2 \cup \dots \cup N_t}\|_2^2 
= \sum_{i=1}^{t} \|[(A^{\top} A - I)v]_{N_i}\|_2^2 
\le t \, \delta_{3k}^2 \, \|v\|_2^2.
\]
Thus, we have
\begin{equation} \label{lemma1-2}
\|(A^{\top} A - I) v\|_2 \le \sqrt{\left\lceil \tfrac{n}{k} \right\rceil} \delta_{3k} \|v\|_2.
\end{equation}
Finally, substituting \eqref{lemma1-1} and \eqref{lemma1-2} into \eqref{lemma1-3} and combining \eqref{lemma1-5} yields the desired bound \eqref{lemma1-4}.
The proof of case (i) is complete.

We now turn to case (ii).
From \eqref{lemma1-7}, we immediately obtain
\begin{equation} \label{lemma1-8}
\|M(\Omega) Z(\Omega) A^{\top} u\|_2 \le \|M(\Omega) Z(\Omega)\|_2 \|A^{\top} u\|_2 = \|\widetilde{M}(\Omega) \widetilde{Z}(\Omega)\|_2 \|A^{\top} u\|_2.
\end{equation}
If $\alpha$ and $\gamma$ are chosen as $\frac{1}{1+\delta_q} \leq \alpha \gamma \le \frac{1}{1-\delta_q}$, then it follows directly that
\[
\|\widetilde{M}(\Omega) \widetilde{Z}(\Omega)\|_2 = \frac{1}{1-\delta_q}.
\]
Next, partition $[N]$ into the $t=\lceil n / k\rceil$ disjoint subsets $N_1, \ldots, N_t$ as in \eqref{lemma1-6}.
Applying Lemma \ref{lemma-rip} (ii) to each block $N_i$ gives
\[
\| A^{\top} u \|_2 = \| (A^{\top} u)_{N_1 \cup \dots \cup N_t} \|_2 \le \sqrt{\left\lceil\tfrac{n}{k}\right\rceil (1+\delta_k)} \| u \|_2.
\]
Substituting these two inequalities into \eqref{lemma1-8} yields \eqref{lemma1-ii}, which completes the proof of case (ii).
\end{proof}

\subsection{Theoretical results for CNHT and CNHTP}
When $x$ is compressible (i.e., not exactly $k$-sparse), let $x_S$ denote the best $k$-term approximation of $x$, where $S=\mathcal{L}_k\left(x\right)$.
Notice that
\begin{equation} \label{linear}
y=A x+\eta = Ax_S + Ax_{\overline{S}} + \eta = Ax_S + \eta',
\end{equation}
where $\eta' = Ax_{\overline{S}} + \eta$.
Under suitable choices of $\lambda$ and $q$, we establish the following theorem which provides an error bound for the iterates generated by the algorithm.
This result further indicates that the algorithms converge to $x_S$ when $\eta'=0$.

\vspace{0.5em}

\begin{theorem} \label{thm-cnt}
Let $y:=Ax+\eta$, where $x$ is $k$-compressible and $S = {\cal L}_k (x)$, and $q \ge k$ be a positive integer.
Let $A \in \mathbb{R}^{m \times n}$ be a matrix satisfying
\begin{equation} \label{thm1-delta}
\delta_{\max \{q,3k\}} < \frac{-1+\sqrt{1+\frac{1}{\varphi}\sqrt{\left\lceil\frac{n}{k}\right\rceil}}}{\sqrt{\left\lceil\frac{n}{k}\right\rceil}}.
\end{equation}
Let $\lambda > 0$ is chosen such that
\begin{equation} \label{thm1-lam}
\left(1+\delta_q\right)\left(1-\frac{\frac{1}{\varphi}-\delta_{3 k}}{1+\sqrt{\left\lceil\tfrac{n}{k}\right\rceil} \delta_{3 k}}\right) < \lambda < \left(1-\delta_q\right)\left(1+\frac{\frac{1}{\varphi}-\delta_{3 k}}{1+\sqrt{\left\lceil\frac{n}{k}\right\rceil} \delta_{3 k}}\right),
\end{equation}
and parameters $\alpha, \gamma$ satisfy $\frac{1}{1+\delta_q} \le \alpha \gamma \le \frac{1}{1-\delta_q}$.
Then the sequence $\{x^{(p)}\}$ generated by the CNHT satisfies
\begin{equation} \label{thm1-convergence}
\|x^{(p+1)} - x_S\|_2 \le \rho_1 \|x^{(p)} - x_S\|_2 + \frac{\lambda \varphi \sqrt{\left\lceil\frac{n}{k}\right\rceil \left(1+\delta_k\right)}}{1-\delta_q}\left\|\eta^{\prime}\right\|_2,
\end{equation}
where $\eta' = Ax_S +\eta$ and 
\begin{equation} \label{thm1-rho}
\rho_1 := \varphi\left(\delta_{3 k}+\Delta\left(\lambda, \delta_q\right)\left(1+\sqrt{\left\lceil\tfrac{n}{k}\right\rceil} \delta_{3 k}\right)\right).
\end{equation}
Here, 
$\Delta (\lambda, \delta_q) $ is given by (\ref{lemma1-10}). 
$\rho_1<1$ is guaranteed under \eqref{thm1-delta} and \eqref{thm1-lam}.
Moreover, the sequence $\{x^{(p)}\}$ converges to $x_S$ when $\eta' = 0$.
\end{theorem}

\begin{proof}
Let $x^{(p+1)}$ be generated by CNHT.
By Lemma \ref{lemma-hard} with $u=x_S$ and $z=u^{(p)}$, we have
\begin{equation} \label{thm1-11}
\|x^{(p+1)} - x_{S}\|_2 = \|{\cal H}_k (u^{(p)}) - x_{S}\|_2 
\le \varphi \|(u^{(p)} - x_{S})_{S^{p+1} \cup S}\|_2,
\end{equation}
where $\varphi=(\sqrt{5}+1) / 2,$ $ S=\textrm{supp} (x_S),$ and $ S^{p+1} = \textrm{supp} ({\cal H}_k(u^p)). $
Using \eqref{alg1-CN} and \eqref{linear}, we have
\[
\begin{aligned}
u^{(p)}-x_{S} & =x^{(p)}-x_{S}+\lambda M(\Omega^p) Z(\Omega^p) A^{\top}(y-A x^{(p)}) \\
&= x^{(p)}-x_{S}+\lambda M(\Omega^p) Z(\Omega^p) A^{\top}(A x_{S}+\eta^{\prime}-A x^{(p)}) \\
&= (I-\lambda M(\Omega^p) Z(\Omega^p) A^{\top} A) (x^{(p)}-x_{S})
+\lambda M(\Omega^p) Z(\Omega^p) A^{\top} \eta'.
\end{aligned}
\]
Applying the triangle inequality leads to
\begin{align}
\|(u^{(p)} - x_{S})_{S^{p+1} \cup S}\|_2
\le& \|((I-\lambda M(\Omega^p) Z(\Omega^p) A^{\top} A) (x^{(p)}-x_{S}))_{S^{p+1} \cup S}\|_2 \nonumber \\
&+ \lambda \|(M(\Omega^p) Z(\Omega^p) A^{\top} \eta')_{S^{p+1} \cup S} \|_2. \label{thm1-5}
\end{align}
By Lemma \ref{lemma1}, the first term on the right-hand side of \eqref{thm1-5} can be bounded as
\begin{align}
&\|((I-\lambda M(\Omega^p) Z(\Omega^p) A^{\top} A) (x^{(p)}-x_{S}))_{S^{p+1} \cup S}\|_2 \nonumber \\
&\le \left(\delta_{3k} + \Delta(\lambda, \delta_q)\left(1 + \sqrt{\left\lceil \tfrac{n}{k} \right\rceil}\, \delta_{3k}\right)\right) \|x^{(p)}-x_{S}\|_2, \label{thm1-10}
\end{align}
where $\Delta(\lambda, \delta_q)$ is given as \eqref{lemma1-10}.
Applying Lemma \ref{lemma1} (ii), the second term on the right-hand side of \eqref{thm1-5} can be bounded as
\begin{equation} \label{thm1-12}
\|(M(\Omega^p) Z(\Omega^p) A^{\top} \eta')_{S^{p+1} \cup S} \|_2 \le \|M(\Omega^p) Z(\Omega^p) A^{\top} \eta' \|_2 \le \frac{\sqrt{\left\lceil\tfrac{n}{k}\right\rceil (1+\delta_k)}}{1-\delta_q} \| \eta' \|_2.
\end{equation}

Substituting \eqref{thm1-10} and \eqref{thm1-12} into \eqref{thm1-5}, together with \eqref{thm1-11}, leads to \eqref{thm1-convergence}.
We now find the conditions guaranteeing $\rho_1<1$, i.e.,
\begin{equation} \label{thm1-14}
\varphi \left(\delta_{3k} + \Delta(\lambda, \delta_q)\left(1 + \sqrt{\left\lceil \tfrac{n}{k} \right\rceil}\, \delta_{3k}\right)\right) < 1.
\end{equation}
It is easy to verify that
\[
\Delta\left(\lambda, \delta_q\right)= : \max \left\{\left|1-\frac{\lambda}{1-\delta_q}\right|,\left|1-\frac{\lambda}{1+\delta_q}\right|\right\} =\begin{cases}1-\frac{\lambda}{1+\delta_q}, & \text{when } \lambda \leq 1-\delta_q^2, \\ \frac{\lambda}{1-\delta_q}-1, & \text{when } \lambda>1-\delta_q^2.\end{cases}
\]
Hence, there are only two cases.

\emph { Case 1:} $\lambda \leq 1-\delta_q^2$. In this case,  $\Delta\left(\lambda, \delta_q\right)=1-\frac{\lambda}{1+\delta_q}$, and thus for this case the inequality  \eqref{thm1-14} becomes
\[
\varphi \left(\delta_{3k} + \left(1-\frac{\lambda}{1+\delta_q}\right)\left(1 + \sqrt{\left\lceil \tfrac{n}{k} \right\rceil}\, \delta_{3k}\right)\right) < 1.
\]
This can be guaranteed provided that $ \lambda $ is chosen such that 
\[
\left(1+\delta_q\right)\left(1-\frac{\frac{1}{\varphi}-\delta_{3 k}}{1+\sqrt{\left\lceil\frac{n}{k}\right\rceil} \delta_{3 k}}\right) < \lambda \leq 1-\delta_q^2.
\]
This range for $ \lambda $ exists if
\begin{equation} \label{thm1-15}
\delta_q < \frac{\frac{1}{\varphi}-\delta_{3 k}}{1+\sqrt{\left\lceil\frac{n}{k}\right\rceil} \delta_{3 k}},
\end{equation}
which is ensured under the condition \eqref{thm1-delta}. Indeed, let us consider the quadratic inequality 
\begin{equation} \label{thm1-2}
h_1(\delta_{\max \{q, 3 k\}}) := \sqrt{\left\lceil\tfrac{n}{k}\right\rceil} \delta_{\max \{q, 3 k\}}^2 + 2 \delta_{\max \{q, 3 k\}} - \frac{1}{\varphi} < 0.
\end{equation}
It is straightforward to verify that $h_1(0)<0$, $h_1(1)>0$, and that $h_1$ is strictly increasing on $(0,1)$.  
Thus, there exists a unique real root of $h_1(\delta_{\max \{q,3k\}})=0$ in $(0,1)$, which is given by the right-hand side of \eqref{thm1-delta}. Thus the inequality \eqref{thm1-2} is ensured  under \eqref{thm1-delta}. This inequality can be rewritten equivalently as 
\[
\delta_{\max \{q, 3 k\}}<\frac{\frac{1}{\varphi}-\delta_{\max \{q, 3 k\}}}{1+\sqrt{\left\lceil\frac{n}{k}\right\rceil} \delta_{\max \{q, 3 k\}}},
\]
which implies (\ref{thm1-15}) by noting that $\delta_q \leq \delta_{\max \{q, 3 k\}}$ and $\delta_{3 k} \leq \delta_{\max \{q, 3 k\}}$.

\emph{Case 2.}  $\lambda > 1-\delta_q^2$. In this case,  $\Delta\left(\lambda, \delta_q\right)=\frac{\lambda}{1-\delta_q}-1$, and hence the inequality \eqref{thm1-14} becomes
\[
\varphi \left(\delta_{3k} + \left(\frac{\lambda}{1-\delta_q}-1\right) \left(1 + \sqrt{\left\lceil\tfrac{n}{k}\right\rceil}\, \delta_{3k}\right)\right) < 1.
\]
By a similar proof to the case 1, the above inequality holds if  $\lambda$ satisfies that 
\[
1-\delta_q^2 \leq \lambda < \left(1-\delta_q\right)\left(1+\frac{\frac{1}{\varphi}-\delta_{3 k}}{1+\sqrt{\left\lceil\tfrac{n}{k}\right\rceil} \delta_{3 k}}\right).
\]
The existence of this interval is ensured under the condition \eqref{thm1-delta}.

Moreover, when $\eta'=0$, \eqref{thm1-convergence} reduces to $\|x^{(p+1)}-x_S\|_2 \leq \rho_1\|x^{(p)}-x_S\|_2$ with $\rho_1 < 1$, implying that $x^{(p+1)}$ converges to $x_S$.
This completes the proof.
\end{proof}

The following lemma establishes a key property associated with the pursuit step.

\begin{lemma} \emph{\cite{bouchot2016hard, zhao2020optimal}} \label{lemma-pursuit}
Let $y:=A u+\eta$, where $u \in \mathbb{R}^n$ is a $k$-sparse vector and $\eta \in \mathbb{R}^m$ denotes the noise.
Consider an arbitrary $k$-sparse vector $v$.
The solution to the projection
\[
z^*=\arg \min _z\left\{\|y-A z\|_2^2: \operatorname{supp}(z) \subseteq \operatorname{supp}(v)\right\}
\]
satisfies that 
\[
\left\|z^*-u\right\|_2 \leq \frac{1}{\sqrt{1-\delta_{2 k}^2}} \|u-v\|_2+\frac{\sqrt{1+\delta_k}}{1-\delta_{2 k}}\|\eta\|_2.
\]
\end{lemma}

\begin{theorem} \label{thm-cntp}
Let $y:=Ax+\eta$, where $x$ is $k$-compressible and $S = {\cal L}_k (x)$, and $q \ge k$ be a positive integer.
Let $A \in \mathbb{R}^{m \times n}$ be a matrix satisfying
\begin{equation} \label{thm-CNTP-q}
\delta_{\max\{q,3k\}} < \frac{-(2 \varphi+1) + \sqrt{(2 \varphi+1)^2+4 \varphi \sqrt{\left\lceil\frac{n}{k}\right\rceil}}}{2 \varphi \sqrt{\left\lceil\frac{n}{k}\right\rceil}}.
\end{equation}
If $\lambda$ is chosen such that
\begin{equation} \label{thm-CNTP-lam}
\left(1+\delta_q\right)\left(1-\frac{\frac{1}{\varphi}\sqrt{1-\delta_{2 k}^2}-\delta_{3 k}}{1 + \sqrt{\left\lceil \tfrac{n}{k} \right\rceil} \delta_{3k}}\right)<\lambda < \left(1-\delta_q\right)\left(1+\frac{\frac{1}{\varphi}\sqrt{1-\delta_{2 k}^2}-\delta_{3 k}}{1+\sqrt{\left\lceil\frac{n}{k}\right\rceil} \delta_{3 k}}\right),
\end{equation}
and parameters $\alpha, \gamma$ satisfy $\frac{1}{1+\delta_q} \le \alpha \gamma \le \frac{1}{1-\delta_q}$.
Then the sequence $\left\{x^{(p)}\right\}$ generated by the CNHTP satisfies that
\begin{equation} \label{thm2-convergence}
\|x^{(p+1)}-x_{S}\|_2 \leq \rho_2 \|x^{(p)}-x_{S}\|_2 + \left( \frac{\lambda \varphi}{1-\delta_q} \sqrt{\frac{\left\lceil\frac{n}{k}\right\rceil\left(1+\delta_k\right)}{1-\delta_{2k}^2}}
+\frac{\sqrt{1+\delta_k}}{1-\delta_{2 k}}\right)\left\|\eta'\right\|_2,
\end{equation}
where $\eta' = Ax_S +\eta$,
\begin{equation} \label{thm2-rho}
\rho_2 := \frac{\varphi }{\sqrt{1-\delta_{2 k}^2}}\left(\delta_{3k} + \Delta(\lambda, \delta_q)\left(1 + \sqrt{\left\lceil \tfrac{n}{k} \right\rceil}\, \delta_{3k}\right)\right),
\end{equation}
and  
$\Delta (\lambda, \delta_q) $ is given by (\ref{lemma1-10}).
$\rho_2 < 1$ is guaranteed under \eqref{thm-CNTP-q} and \eqref{thm-CNTP-lam}.
Moreover, the sequence $\{x^{(p)}\}$ converges to $x_S$ when $\eta' = 0$.
\end{theorem}

\begin{proof}
Let $x^{(p+1)}$ be generated by CNHTP.
Applying Lemma \ref{lemma-pursuit} with $u = x_{S}$, $v = {\cal H}_k (u^{(p)})$, and $\eta = \eta'$, we have
\[
\|x^{(p+1)} - x_{S}\|_2 \le \frac{1}{\sqrt{1-\delta_{2k}^2}} \|{\cal H}_k (u^{(p)}) - x_{S}\|_2 + \frac{\sqrt{1+\delta_k}}{1-\delta_{2 k}}\|\eta'\|_2.
\]
From Theorem \ref{thm-cnt}, it follows that
\[
\|{\cal H}_k (u^{(p)}) - x_{S}\|_2 \le \rho_1 \|x^{(p)} - x_S\|_2 + \frac{\lambda \varphi \sqrt{\left\lceil\frac{n}{k}\right\rceil \left(1+\delta_k\right)}}{1-\delta_q}\left\|\eta'\right\|_2,
\]
where $\rho_1$ is given by \eqref{thm1-rho}.
Combining these two inequalities leads to \eqref{thm2-convergence}.
We now verify that $\rho_2 < 1$, where $\rho_2$ is given by \eqref{thm2-rho}, under the assumption of the theorem.

\emph{Case 1.}   $\lambda \leq 1-\delta_q^2$.  In this case,  $\Delta\left(\lambda, \delta_q\right)=1-\frac{\lambda}{1+\delta_q}$.
Thus, $\rho_2<1$ is written as
\[
\frac{\varphi}{\sqrt{1-\delta_{2 k}^2}} \left(\delta_{3k} + \left(1-\frac{\lambda}{1+\delta_q}\right)\left(1 + \sqrt{\left\lceil \tfrac{n}{k} \right\rceil}\, \delta_{3k}\right)\right) <1,
\]
which can be ensured if $ \lambda $ is chosen such that 
\[
\left(1+\delta_q\right)\left(1-\frac{\frac{1}{\varphi}\sqrt{1-\delta_{2 k}^2}-\delta_{3 k}}{1 + \sqrt{\left\lceil \tfrac{n}{k} \right\rceil} \delta_{3k}}\right) < \lambda \leq 1-\delta_q^2.
\]
The existence of this range is guaranteed by 
\begin{equation} \label{thm2-2}
\delta_q < \frac{\frac{1}{\varphi}\sqrt{1-\delta_{2 k}^2}-\delta_{3 k}}{1 + \sqrt{\left\lceil \tfrac{n}{k} \right\rceil} \delta_{3k}},
\end{equation}
which holds under \eqref{thm-CNTP-q}. To verify this, let us consider the following quadratic inequality in $\delta_{\max \{q, 3 k\}}: $
\begin{equation} \label{Q666D}
h_2(\delta_{\max \{q, 3 k\}}) := \varphi \sqrt{\left\lceil\tfrac{n}{k}\right\rceil} \delta_{\max \{q, 3 k\}}^2+(2 \varphi+1) \delta_{\max \{q, 3 k\}}-1<0.
\end{equation}
It is straightforward to verify that the equation $h_2 (\delta_{\max \{q, 3 k\}})=0$ has a unique real root in the interval $(0,1)$, and that $h_2$ is strictly increasing over this interval. This real root is given by the right-hand side of \eqref{thm-CNTP-q}. Thus the condition \eqref{thm-CNTP-q} guarantees the inequality \eqref{Q666D}, which can be written equivalently as
\begin{equation} \label{thm2-1}
\delta_{\max \{q, 3 k\}}<\frac{\frac{1}{\varphi}\left(1-\delta_{\max \{q, 3 k\}}\right)-\delta_{\max \{q, 3 k\}}}{1+\sqrt{\left\lceil\frac{n}{k}\right\rceil} \delta_{\max \{q, 3 k\}}}.
\end{equation}
This further implies (\ref{thm2-2}) by simply noting that  $\delta_{2k} \le \delta_{3k}$, $\delta_q \le \delta_{\max \{q, 3 k\}}$, $\delta_{3k} \le \delta_{\max \{q, 3 k\}}$, and $1-\delta_{\max \{q, 3 k\}} < \sqrt{1-\delta_{\max \{q, 3 k\}}^2}$ since $\delta_{\max \{q, 3 k\}} \in(0,1)$.

\emph{Case 2.} $\lambda > 1-\delta_q^2$. In this case, $\Delta\left(\lambda, \delta_q\right)=\frac{\lambda}{1-\delta_q}-1$, and then $\rho_2<1$ is written as
\[
\frac{\varphi}{\sqrt{1-\delta_{2 k}^2}} \left(\delta_{3k} + \left(\frac{\lambda}{1-\delta_q}-1\right)\left(1 + \sqrt{\left\lceil \tfrac{n}{k} \right\rceil}\, \delta_{3k}\right)\right) < 1.
\]
By a similar analysis to case 1, the above inequality holds if $\lambda$ is chosen such that  
\[
1-\delta_q^2<\lambda<\left(1-\delta_q\right)\left(1+\frac{\frac{1}{\varphi}\sqrt{1-\delta_{2 k}^2}-\delta_{3 k}}{1+\sqrt{\left\lceil\frac{n}{k}\right\rceil} \delta_{3 k}}\right).
\]
The existence of this interval is ensured  \eqref{thm-CNTP-q}.

Moreover, when $\eta'=0$, \eqref{thm2-convergence} reduces to $\|x^{(p+1)}-x_S\|_2 \leq \rho_2 \|x^{(p)}-x_S\|_2$ with $\rho_2 < 1$, implying that $x^{(p+1)}$ converges to $x_S$.
The proof is complete.
\end{proof}

\subsection{Theoretical results for CNOT and CNOTP}
We now establish the main results for the CNOT and CNOTP algorithms.
The following lemma concerning the relaxed optimal $k$-thresholding step follows directly from combining inequalities (27) and (28) in \cite{meng2022partial}.
\vspace{0.5em}

\begin{lemma} \cite{zhao2020optimal, meng2022partial}
Let $y:=A x+\eta$, which can also be expressed as $y = Ax_S + \eta'$ with $\eta' = Ax_{\overline{S}} + \eta$.
Let $u^{(p)}$ and $w^{(p)}$ be defined as in Algorithm \ref{CNOTP}. Then
\begin{equation} \label{lemma7-1}
\|{\cal H}_k (w^{(p)} \otimes u^{(p)}) - x_{S}\|_2 \le 3 \varphi \sqrt{\frac{1+\delta_k}{1-\delta_{2 k}}} \|u^{(p)}-x_{S}\|_2 + \frac{2 \varphi}{\sqrt{1-\delta_{2 k}}}\|\eta'\|_2,
\end{equation}
where $\varphi = \frac{\sqrt{5}+1}{2}$.
\end{lemma}

\begin{theorem}\label{thm-cnot}
Let $y:=Ax+\eta$, where $x$ is $k$-compressible and $S = {\cal L}_k (x)$, and $q \ge k$ be a positive integer.
Let $A \in \mathbb{R}^{m \times n}$ be a matrix satisfying
\begin{equation} \label{thm3-delta}
\delta_{\max \{q, 3 k\}} < \beta^*,
\end{equation}
where $\beta^*$ is the unique real root of
\[
9 \varphi^2\left[\left\lceil\tfrac{n}{k}\right\rceil \beta^5+\left(\left\lceil\tfrac{n}{k}\right\rceil+4 \sqrt{\left\lceil\tfrac{n}{k}\right\rceil}\right) \beta^4+4\left(\sqrt{\left\lceil\tfrac{n}{k}\right\rceil}+1\right) \beta^3+4 \beta^2\right]+\beta-1=0 .
\]
If $\lambda$ is chosen such that
\begin{equation} \label{thm3-lam}
\left(1+\delta_q\right)\left(1-\frac{\frac{1}{3 \varphi} \sqrt{\frac{1-\delta_{2k}}{1+\delta_{k}}} -\delta_{3 k}}{1+\sqrt{\left\lceil\frac{n}{k}\right\rceil} \delta_{3 k}}\right) < \lambda < \left(1-\delta_q\right)\left(1+\frac{\frac{1}{3 \varphi} \sqrt{\frac{1-\delta_{2k}}{1+\delta_{k}}} -\delta_{3 k}}{1+\sqrt{\left\lceil\frac{n}{k}\right\rceil} \delta_{3 k}}\right),
\end{equation}
and parameters $\alpha, \gamma$ satisfy $\frac{1}{1+\delta_q} \le \alpha \gamma \le \frac{1}{1-\delta_q}$.
Then the sequence $\{x^{(p)}\}$ generated by the CNOT satisfies
\begin{equation} \label{thm3-convergence}
\|x^{(p+1)}-x_{S}\|_2 \leq \rho_3 \|x^{(p)}-x_{S}\|_2 
+ \tau_3 \|\eta'\|_2,
\end{equation}
where $\eta' = Ax_S +\eta$,
\begin{equation} \label{thm3-rho}
\rho_3 := 3 \varphi \sqrt{\frac{1+\delta_k}{1-\delta_{2 k}}} \left(\delta_{3 k}+\Delta\left(\lambda, \delta_q\right)\left(1+\sqrt{\left\lceil\tfrac{n}{k}\right\rceil} \delta_{3 k}\right)\right),
\end{equation}
\begin{equation} \label{thm3-tau}
\tau_3 := \frac{3 \lambda \varphi (1+\delta_k) \sqrt{\left\lceil\tfrac{n}{k}\right\rceil}}{(1-\delta_q)\sqrt{1-\delta_{2 k}}} + \frac{2\varphi}{\sqrt{1-\delta_{2 k}}}.
\end{equation} Here, 
$\Delta (\lambda, \delta_q) $ is given by (\ref{lemma1-10}).
$\rho_3<1$ is guaranteed under \eqref{thm3-delta} and \eqref{thm3-lam}.
Moreover, the sequence $\{x^{(p)}\}$ converges to $x_S$ when $\eta' = 0$.
\end{theorem}

\begin{proof}
The iterate $x^{(p+1)}$ is generated by CNOT.
By \eqref{lemma7-1}, we have a bound relating $\|\mathcal{H}_k(w^{(p)} \otimes u^{(p)})-x_S\|_2$ and $\left\|u^{(p)}-x_S\right\|_2$, which is directly applicable in the analysis.
Recall the upper bound on $\|u^{(p)}-x_S\|_2$ from the proof of Theorem \ref{thm-cnt}, which is given by
\[
\|u^{(p)}-x_{S}\|_2 \leq \left(\delta_{3 k}+\Delta\left(\lambda, \delta_q\right)\left(1+\sqrt{\left\lceil\tfrac{n}{k}\right\rceil} \delta_{3 k}\right)\right) \|x^{(p)}-x_{S}\|_2+ \frac{\lambda \sqrt{\left\lceil\frac{n}{k}\right\rceil\left(1+\delta_k\right)}}{1-\delta_q} \|\eta^{\prime}\|_2,
\]
where $\Delta (\lambda, \delta_q)$ is given in (\ref{lemma1-10}).
Substituting this into \eqref{lemma7-1} leads to \eqref{thm3-convergence}.
We now find the range for $\lambda$ to ensure that $\rho_3<1$, where $\rho_3$ is given by \eqref{thm3-rho}.

\emph{Case 1.}   $\lambda \leq 1-\delta_q^2$. In this case, $\rho_3<1$ is written as 
\[
3 \varphi \sqrt{\frac{1+\delta_k}{1-\delta_{2 k}}} \left(\delta_{3 k}+ \left(1-\frac{\lambda}{1+\delta_q}\right)\left(1+\sqrt{\left\lceil\tfrac{n}{k}\right\rceil} \delta_{3 k}\right)\right) < 1,
\]
which is ensured if $\lambda$ satisfies that
\[
\left(1+\delta_q\right)\left(1-\frac{\frac{1}{3 \varphi} \sqrt{\frac{1-\delta_{2k}}{1+\delta_{k}}} -\delta_{3 k}}{1+\sqrt{\left\lceil\frac{n}{k}\right\rceil} \delta_{3 k}}\right) < \lambda \le 1-\delta_q^2.
\]
The existence of this interval is guaranteed by
\begin{equation} \label{thm3-3}
\delta_q < \left(\frac{1}{3 \varphi} \sqrt{\frac{1-\delta_{2k}}{1+\delta_{k}}} -\delta_{3 k}\right) / \left(1+\sqrt{\left\lceil\tfrac{n}{k}\right\rceil} \delta_{3 k}\right),
\end{equation}
which can be ensured by
\begin{equation} \label{thm3-2}
\delta_{\max\{q,3k\}} < \left(\frac{1}{3 \varphi} \sqrt{\frac{1-\delta_{\max\{q,3k\}}}{1+\delta_{\max\{q,3k\}}}}-\delta_{\max \{q, 3 k\}}\right) / \left(1+\sqrt{\left\lceil\tfrac{n}{k}\right\rceil} \delta_{\max \{q, 3 k\}}\right)
\end{equation}
by using $\delta_k \le \delta_{2k} \le \delta_{3k} \le \delta_{\max \{q, 3 k\}}$ and $\delta_q \le \delta_{\max \{q, 3 k\}}$.
For notation convenience, denote $\zeta=\delta_{\max \{q, 3 k\}}$.
We now point out that the inequality \eqref{thm3-2} is guaranteed under the condition \eqref{thm3-delta}. In fact, \eqref{thm3-2} is  equivalent to $h_3(\zeta) < 0$, where
\[
h_3(\zeta) := 9 \varphi^2\left[\left\lceil\tfrac{n}{k}\right\rceil \zeta^5+\left(\left\lceil\tfrac{n}{k}\right\rceil+4 \sqrt{\left\lceil\tfrac{n}{k}\right\rceil}\right) \zeta^4+4\left(\sqrt{\left\lceil\tfrac{n}{k}\right\rceil}+1\right) \zeta^3+4 \zeta^2\right]+\zeta-1.
\]
Since $h_3(0)=-1<0, h_3(1)>0$, and $h_3(\zeta)$ is monotonically increasing on $(0,1)$, there exists a unique root $\beta^* \in(0,1)$ satisfying $h_3(\beta^*) = 0$.
Therefore, \eqref{thm3-2} can be guaranteed by $\delta_{\max \{q, 3 k\}} < \beta^*$, which is exactly the condition \eqref{thm3-delta}.

\emph{Case 2.}  $\lambda > 1-\delta_q^2$.  In this case, the inequality $\rho_3<1$ is written as 
\[
3 \varphi \sqrt{\frac{1+\delta_k}{1-\delta_{2 k}}} \left(\delta_{3 k}+ \left(\frac{\lambda}{1-\delta_q}-1\right)\left(1+\sqrt{\left\lceil\tfrac{n}{k}\right\rceil} \delta_{3 k}\right)\right) < 1,
\]
which is guaranteed by choosing $\lambda$ such that
\[
1-\delta_q^2 < \lambda < \left(1-\delta_q\right)\left(1+\frac{\frac{1}{3 \varphi} \sqrt{\frac{1-\delta_{2k}}{1+\delta_{k}}} -\delta_{3 k}}{1+\sqrt{\left\lceil\tfrac{n}{k}\right\rceil} \delta_{3 k}}\right) .
\]
By a similar analysis to Case 1, 
the existence of such an interval is also guaranteed by \eqref{thm3-delta}.

Moreover, when $\eta' = 0$, \eqref{thm3-convergence} reduces to $\|x^{(p+1)}-x_S\|_2 \leq \rho_3 \|x^{(p)}-x_S\|_2$ with $\rho_3 < 1$, implying that $x^{(p+1)}$ converges to $x_S$.
Hence the theorem follows.
\end{proof}

\begin{theorem} \label{thm-cnotp}
Let $y:=Ax+\eta$, where $x$ is $k$-compressible and $S = {\cal L}_k (x)$, and $q \ge k$ be a positive integer.
Let $A \in \mathbb{R}^{m \times n}$ be a matrix satisfying
\begin{equation} \label{thm-cnotp-delta}
\delta_{\max \{q, 3 k\}} < \frac{-(1+6 \varphi)+\sqrt{(1+6 \varphi)^2+12 \varphi \sqrt{\left\lceil\tfrac{n}{k}\right\rceil}}}{6 \varphi \sqrt{\left\lceil\tfrac{n}{k}\right\rceil}}.
\end{equation}
If $\lambda$ is chosen such that
\begin{equation} \label{thm-cnotp-lam}
\left(1+\delta_q\right)\left(1-\frac{\frac{1}{3 \varphi} (1-\delta_{2 k})-\delta_{3 k}}{1+\sqrt{\left\lceil\tfrac{n}{k}\right\rceil} \delta_{3 k}}\right) < \lambda < \left(1-\delta_q\right)\left(1+\frac{\frac{1}{3 \varphi} (1-\delta_{2 k})-\delta_{3 k}}{1+\sqrt{\left\lceil\tfrac{n}{k}\right\rceil} \delta_{3 k}}\right),
\end{equation}
and parameters $\alpha, \gamma$ satisfy $\frac{1}{1+\delta_q} \le \alpha \gamma \le \frac{1}{1-\delta_q}$.
Then the sequence $\{x^{(p)}\}$ generated by the CNOTP satisfies
\begin{equation} \label{thm4-convergence}
\|x^{(p+1)}-x_{S}\|_2 \leq \rho_4 \|x^{(p)}-x_{S}\|_2 + \tau_4 \|\eta'\|_2,
\end{equation}
where $\eta' = Ax_S +\eta$,
\begin{equation} \label{thm4-rho}
\rho_4 := \frac{3 \varphi}{1-\delta_{2 k}}\left(\delta_{3 k}+ \Delta\left(\lambda, \delta_q\right) \left(1+\sqrt{\left\lceil\tfrac{n}{k}\right\rceil} \delta_{3 k}\right)\right),
\end{equation}
\[
\tau_4 := \frac{3 \lambda \varphi\left(1+\delta_k\right) \sqrt{\left\lceil\frac{n}{k}\right\rceil}}{(1-\delta_q) (1-\delta_{2k}) \sqrt{1+\delta_{2k}}} + \frac{2 \varphi}{(1-\delta_{2k}) \sqrt{1+\delta_{2k}}} + \frac{\sqrt{1+\delta_k}}{1-\delta_{2 k}}.
\]
Here, 
$\Delta (\lambda, \delta_q) $ is given by (\ref{lemma1-10}).
$\rho_4<1$ is guaranteed under \eqref{thm-cnotp-delta} and \eqref{thm-cnotp-lam}.
Moreover, the sequence $\{x^{(p)}\}$ converges to $x_S$ when $\eta' = 0$.
\end{theorem}

\begin{proof}
The CNOTP algorithm integrates the CNOT and  a projection step.
Applying Lemma \ref{lemma-pursuit} to the projection step yields
\begin{equation} \label{thm4-1}
\|x^{(p+1)} - x_{S}\|_2 
\le \frac{1}{\sqrt{1-\delta_{2 k}^2}} \|{\cal H}_k (w^{(p)} \otimes u^{(p)}) - x_{S}\|_2 + \frac{\sqrt{1+\delta_k}}{1-\delta_{2 k}}\|\eta'\|_2.
\end{equation}
From the analysis of CNOT, the following bound holds (see \eqref{thm3-convergence}):
\begin{align}
\|\mathcal{H}_k (w^{(p)} \otimes u^{(p)})-x_{S}\|_2 
\leq \rho_3 \|x^{(p)}-x_{S}\|_2 + \tau_3 \|\eta'\|_2, \label{thm4-2}
\end{align}
where $\rho_3$ and $\tau_3$ are given in \eqref{thm3-rho} and \eqref{thm3-tau}, respectively.
Combining \eqref{thm4-1} and \eqref{thm4-2} and using $\delta_k \le \delta_{2k}$ lead to
\[
\begin{aligned}
\|x^{(p+1)} - x_{S}\|_2
\le \frac{\rho_3}{\sqrt{1-\delta_{2 k}^2}} \|x^{(p)}-x_S\|_2 + \left(\frac{\tau_3}{\sqrt{1-\delta_{2 k}^2}}+\frac{\sqrt{1+\delta_{k}}}{1-\delta_{2 k}}\right) \|\eta'\|_2,
\end{aligned}
\]
which gives the desired result \eqref{thm4-convergence} with $\rho_4=\frac{\rho_3}{\sqrt{1-\delta_{2 k}^2}}$, as in \eqref{thm4-rho}.
We now derive the condition for $\rho_4<1$.

\emph{Case 1.} $\lambda \leq 1-\delta_q^2$.  We see that $\rho_4<1$ is written as
\[
\frac{3 \varphi}{1-\delta_{2 k}}\left(\delta_{3 k}+\left(1-\frac{\lambda}{1+\delta_q}\right)\left(1+\sqrt{\left\lceil\tfrac{n}{k}\right\rceil} \delta_{3 k}\right)\right)<1.
\]
This can be ensured if  $\lambda$ is chosen such that 
\[
\left(1+\delta_q\right)\left(1-\frac{\frac{1}{3 \varphi} (1-\delta_{2 k})-\delta_{3 k}}{1+\sqrt{\left\lceil\tfrac{n}{k}\right\rceil} \delta_{3 k}}\right) < \lambda \le 1-\delta_q^2,
\]
such a range for $ \lambda$ exists if 
\begin{equation} \label{thm4-3}
\delta_q < \frac{\frac{1}{3 \varphi}\left(1-\delta_{2 k}\right)-\delta_{3 k}}{1+\sqrt{\left\lceil\frac{n}{k}\right\rceil} \delta_{3 k}},
\end{equation}
We now point out that (\ref{thm4-3}) holds under the assumption (\ref{thm-cnotp-delta}).
Since $\delta_{2k} \le \delta_{3k} \le \delta_{\max \{q, 3 k\}}$, and $\delta_q \le \delta_{\max \{q, 3 k\}}$, to ensure  (\ref{thm4-3}) it is sufficient to guaranteed that 
\[
\delta_{\max \{q, 3 k\}}<\frac{\frac{1}{3 \varphi}\left(1-\delta_{\max \{q, 3 k\}}\right)-\delta_{\max \{q, 3 k\}}}{1+\sqrt{\left\lceil\frac{n}{k}\right\rceil} \delta_{\max \{q, 3 k\}}},
\]
which is equivalent to $h_4(\delta_{\max \{q, 3 k\}}) < 0$, where
\[
h_4(\delta_{\max \{q, 3 k\}}) := 3 \varphi \sqrt{\left\lceil\tfrac{n}{k}\right\rceil} \delta_{\max \{q, 3 k\}}^2+(6 \varphi+1) \delta_{\max \{q, 3 k\}}-1.
\]
The function $h_4$ is monotonically increasing over $(0,1)$ with $h_4(0) = -1 < 0$ and $h_4(1) > 0$.
Hence, there exists a unique real root of $h_4=0$ in the interval $(0,1)$. This unique positive root is exactly the right-hand side of (\ref{thm-cnotp-delta}). Thus the inequality above holds precisely under the condition \eqref{thm-cnotp-delta}.

\emph{Case 2.}  $\lambda > 1-\delta_q^2$. In this case, the inequality $\rho_4 < 1$ becomes
\[
\frac{3 \varphi}{1-\delta_{2 k}}\left(\delta_{3 k}+\left(\frac{\lambda}{1-\delta_q}-1\right)\left(1+\sqrt{\left\lceil\tfrac{n}{k}\right\rceil} \delta_{3 k}\right)\right)<1
\]
which is guaranteed if $\lambda$ satisfies 
\[
1-\delta_q^2<\lambda<\left(1-\delta_q\right)\left(1+\frac{\frac{1}{3 \varphi} (1-\delta_{2 k})-\delta_{3 k}}{1+\sqrt{\left\lceil\tfrac{n}{k}\right\rceil} \delta_{3 k}}\right).
\]
By a similar proof to Case 1, it is easy to verify that the above range for $\lambda$ exists  under the condition \eqref{thm-cnotp-delta}. Moreover, when $\eta' = 0$, \eqref{thm4-convergence} reduces to $\|x^{(p+1)}-x_S\|_2 \leq \rho_4 \|x^{(p)}-x_S\|_2$ with $\rho_4 < 1$, implying that $x^{(p+1)}$ converges to $x_S$.
This completes the proof.
\end{proof}

\subsection{Computational complexity}
The proposed algorithms consist of three primary parts: i) computing the compressed Newton direction, ii) performing thresholding operations, and iii) projecting onto a $k$-dimensional subspace.
Computation of the gradient $\nabla f(x) = A^{\top}(A x - y)$ requires a matrix–vector multiplication with complexity $O(mn)$, while applying the hard thresholding operator ${\cal H}_q$ requires $O(n \log k)$ flops.
The compressed Newton search direction can be obtained by solving
\[
A_{\Omega^{p}}^{\top} A_{\Omega^{p}} (\mathbf{d}_{\mathrm{CN}})_{\Omega^p}= A_{\Omega^{p}}^{\top} (y-A x^{(p)}),
\]
which requires $O\left(q^3\right)$ using Cholesky decomposition of $A_{\Omega^{p}}^{\top} A_{\Omega^{p}}$ along with forward and backward substitutions.
Hence, forming the compressed Newton direction has a total complexity of $O(mn + q^3)$.
For CNHT, the overall per-iteration complexity is $O(mn + q^3)$, as the second application of ${\cal H}_k$ to a vector has a negligible cost of $O(n\log k)$.
For CNOT, solving the relaxed optimal thresholding via an interior-point method requires $O(n^{3.5} L)$ operations \cite{zhao2021analysis}, resulting in a per-iteration complexity of $O(mn + n^{3.5} L)$ since $q \ll n$ in general.
The pursuit step in CNHTP and CNOTP is an orthogonal projection onto a $k$-dimensional subspace, costing $O(k^3)$.
Thus, the per-iteration complexities of CNHTP and CNOTP are $O(mn+q^3+k^3)$ and $O(mn + n^{3.5} L)$, respectively.

Table \ref{complexity} summarizes the computational complexities of the proposed algorithms and other popular thresholding-based methods.

\begin{table}[h]
\centering
\begin{tabular}{|l|l|c|}
\hline
Algorithms       & Computational Complexity              & References                        \\ \hline
IHT             & $O\left(m n\right)$                   &         \cite{blumensath2009iterative}                          \\ \hline
HTP, CoSaMP, SP & $O\left(m n+m^3\right)$               &       \cite{dai2009subspace, needell2009cosamp, foucart2011hard, zhao2021analysis}                            \\ \hline
NSIHT, NSHTP           & $O\left(m n^2+n^3\right)$         & \cite{meng2020newton}             \\ \hline
ROT$_w$         &   $O\left(m^3+m n+n^{3.5} L\right)$ & \cite{zhao2021analysis}  \\ \hline
PGROTP          & $O\left(m n+(q+k)^{3.5} L+k^3\right)$ & \cite{meng2022partial}  \\ \hline
CNHTP          &   $O\left(m n+q^3+k^3\right)$  &  /   \\ \hline
CNOTP          &   $O(mn+n^{3.5} L)$  &   /   \\ \hline
\end{tabular}
\caption{Comparison of the per-iteration computational complexity of a few algorithms.} \label{complexity}
\end{table}

\begin{remark} \label{remark-q}
The integer $q$ is crucial in computing the search direction.
It should be large enough to capture sufficient gradient information, while still ensuring that a valid Newton direction can be computed within the chosen subspace.
If $A$ satisfies $\delta_k<1$, then choosing $q \geq k$ ensures that the Hessian restricted to any $q$-dimensional subspace is invertible, allowing the computation of a Newton direction.
However, setting $q$ too large increases computational cost and may lead to a subspace in which the standard Newton direction is undefined.
Therefore, $k \leq q \leq 2k$ is a reasonable range for $q$, which aligns with the use of $\delta_{2k}$ in the theoretical analysis of the algorithms.
\end{remark}

\section{Numerical Investigations} \label{numerical}
In this section, we use sparse signal recovery as specific examples of sparse optimization problems.
We conduct several types of numerical experiments on random sparse signal recovery.
The positions of the nonzero entries of the synthetic sparse vector $x$ follow a uniform distribution, and their values follow the standard normal distribution $\mathcal{N}(0,1)$.
All convex problems are solved using CVX \cite{grant2020cvx} with the MOSEK solver.
We evaluate the success frequency of the proposed algorithms in solving random sparse optimization problems with different choices of parameters, and we also evaluate the phase transition properties of our algorithms.

\subsection{Sensitivity to parameter choice}

\begin{figure}
\centering
\subfigure[CNHTP]{
\includegraphics[width=0.48\textwidth]{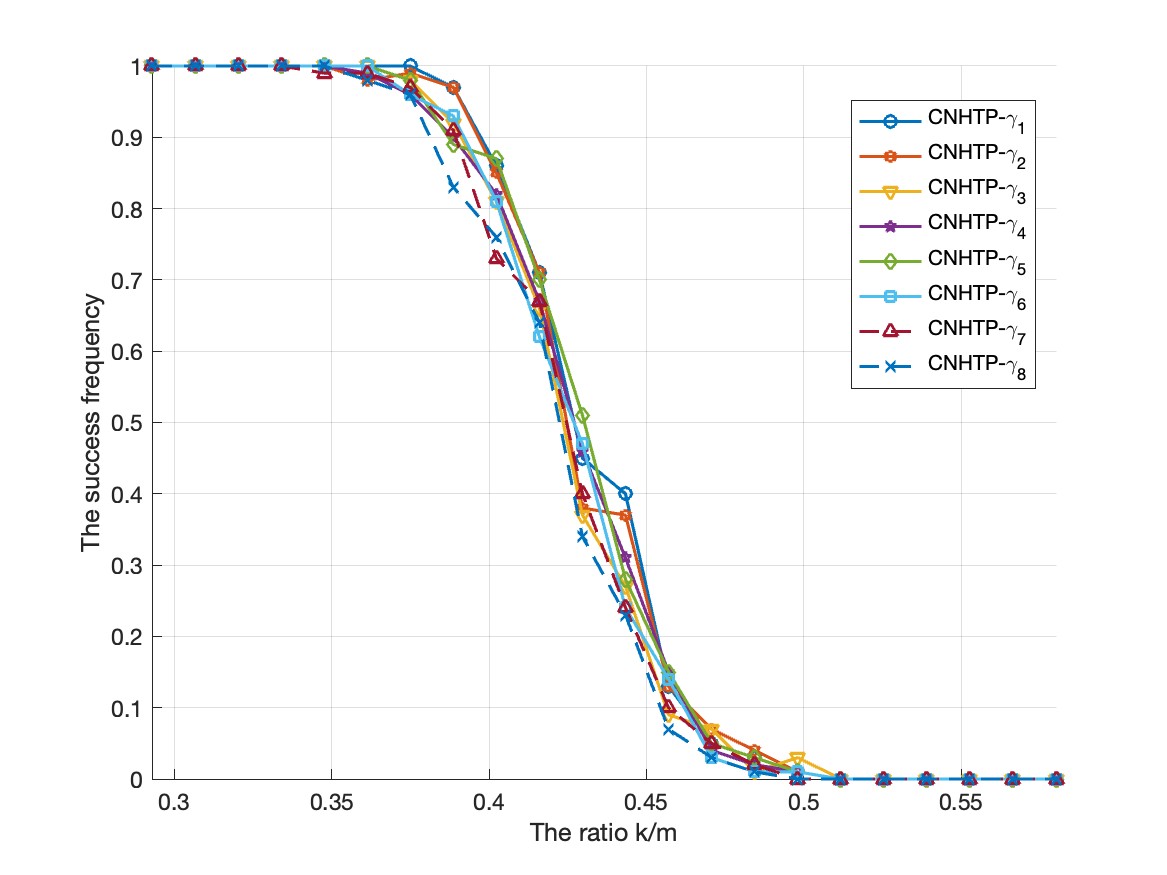}}
\subfigure[CNOTP]{
\includegraphics[width=0.48\textwidth]{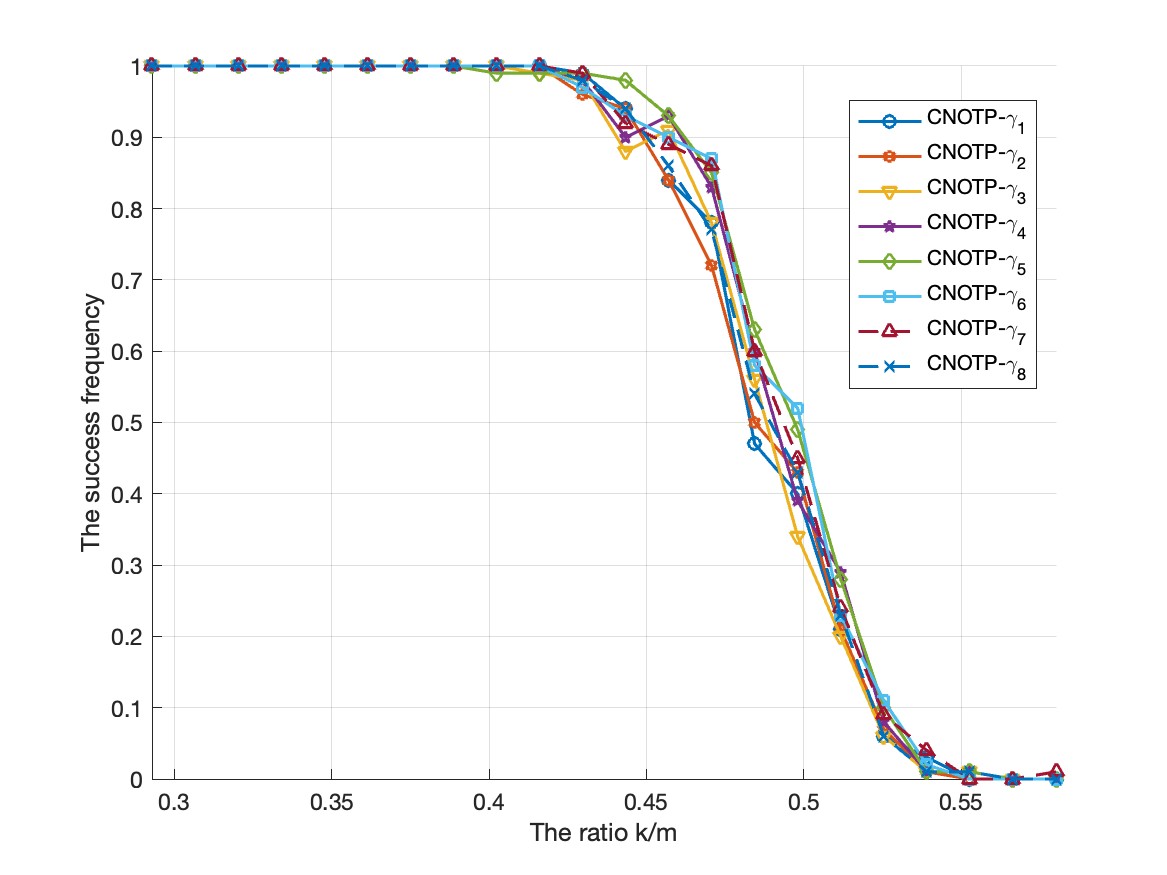}}
\caption{Performance comparison of CNHTP and CNOTP with $\lambda=1$, $q=k$, $\alpha=1$, and different values of $\gamma$ ($\gamma_1=0$, $\gamma_2=0.001$, $\gamma_3=0.01$, $\gamma_4=0.1$, $\gamma_5=0.3$, $\gamma_6=0.5$, $\gamma_7=0.7$, $\gamma_8=1$).}
\label{fig-parameter-gamma}
\end{figure}

We first investigate the success frequency for signal recovery via CNHTP and CNOTP algorithms under different selections of the parameters $(q, \lambda, \alpha, \gamma)$.
The matrix $A \in \mathbb{R}^{512 \times 1024}$ is generated with i.i.d. Gaussian entries from $\mathcal{N}(0,1 / m)$ with $m$=512, which satisfies the RIP with a high probability \cite{candes2005decoding}.
Observations are given by $y=A x+10^{-5} \eta$, where $\eta \sim \mathcal{N}(0, I)$ is standard Gaussian noise.
The algorithm terminates when the criterion
\begin{equation} \label{stop-criteria}
\frac{\left\|x^{(p)}-x\right\|_2}{\left\|x\right\|_2} \le 10^{-3}
\end{equation}
is satisfied or when the maximum of 30 iterations has been executed.
The algorithm is said to be successful if the found solution $x^{(p)}$ satisfies \eqref{stop-criteria}.
We vary the sparsity level $k$ from 150 to 300 in steps of 7, yielding 22 test points.
For each sparsity level $k$, 100 random problem instances are generated to compute the success frequency.

Three sets of experiments are conducted to evaluate the performance of the algorithm. 
The first set examines the effect of varying $\gamma$ while fixing $\lambda=1$, $\alpha=1$, and $q=k$.
Specifically, $\gamma$ is chosen respectively as $\gamma_1=0$, $\gamma_2=0.001$, $\gamma_3=0.01$, $\gamma_4=0.1$, $\gamma_5=0.3$, $\gamma_6=0.5$, $\gamma_7=0.7$, and $\gamma_8=1$.
The results for CNHTP and CNOTP are presented in Fig. \ref{fig-parameter-gamma}, which indicates that both algorithms are not particularly sensitive to the choice of $\gamma$.
Thus the subsequent experiments simply adopt $\gamma=0.1$ unless otherwise specified.

The second experiment investigates the sensitivity of the algorithm according to different choice of $q$ but with fixed $\lambda=1$, $\gamma=0.01$, $\alpha=1$, specifically, $q_1=k, q_2=\lceil 1.25 k\rceil, q_3=\lceil 1.5 k\rceil$, and $q_4=2 k$, where $k$ denotes the sparsity level of $x$.
The results are presented in Fig. \ref{fig-parameter1}.
It shows that both CNHTP and CNOTP have a slight higher success frequency when $q=k$.
This choice of $q$ is consistent with the analysis in Remark \ref{remark-q}, which suggests selecting a support size close to the true sparsity level to preserve sufficient information in the search direction.

\begin{figure}
\centering
\subfigure[CNHTP]{
\includegraphics[width=0.48\textwidth]{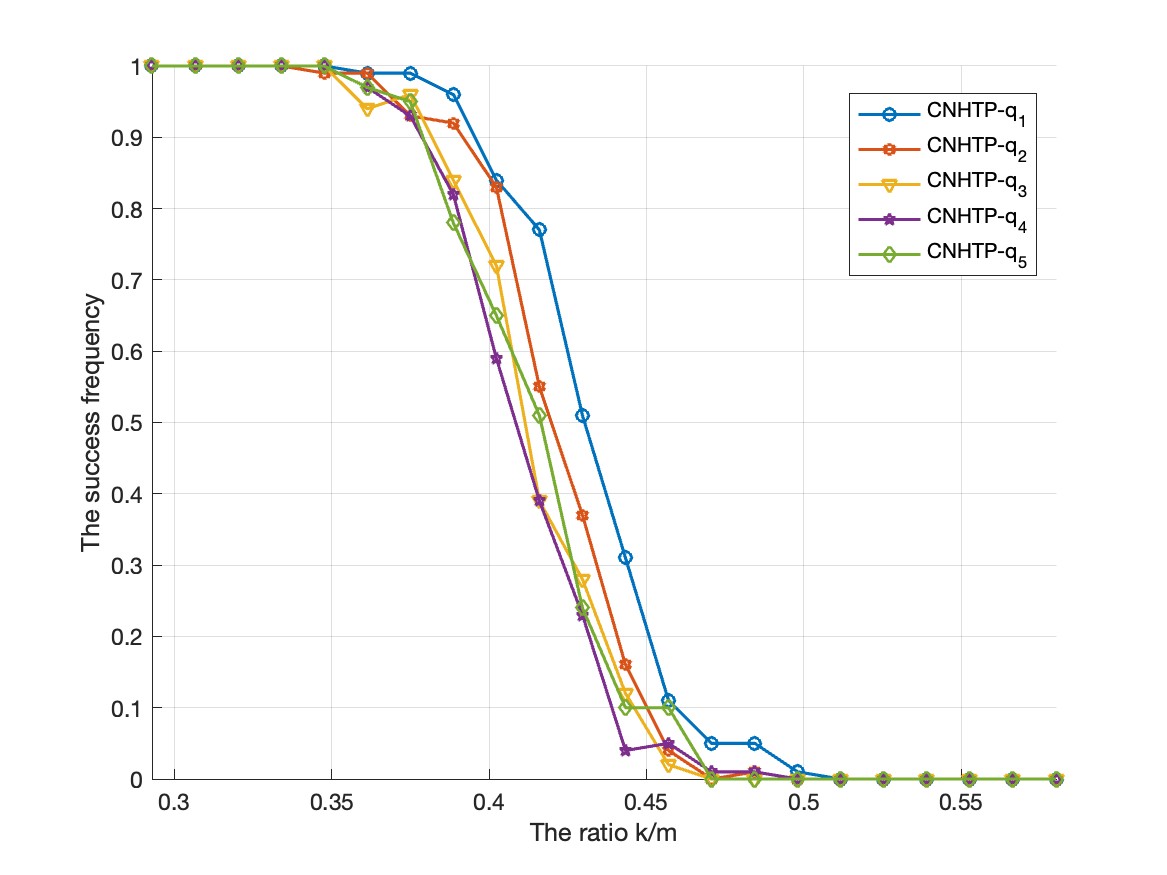}}
\subfigure[CNOTP]{
\includegraphics[width=0.48\textwidth]{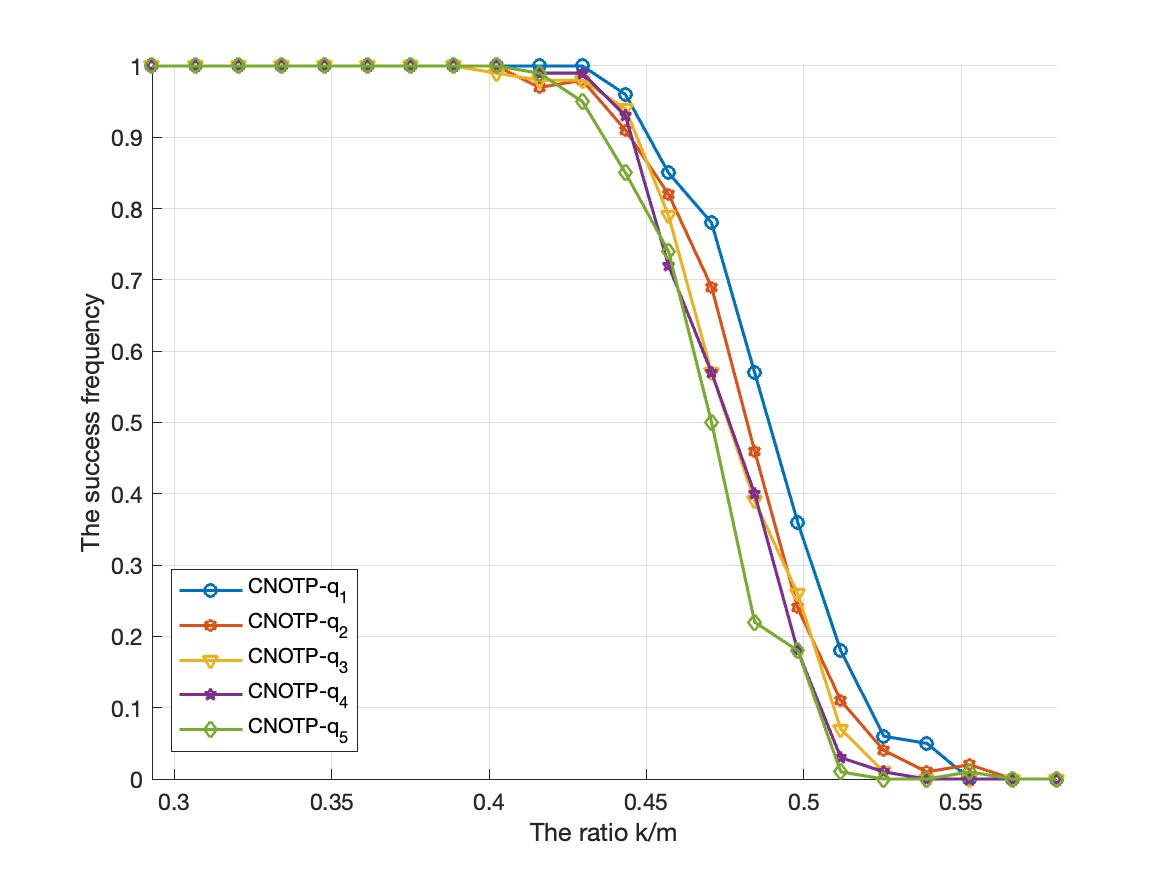}}
\caption{Performance comparison of CNHTP and CNOTP with $\lambda=1$, $\alpha=1$, $\gamma=0.01$ and different values of $q$ ($q_1=k, q_2=\lceil 1.25 k\rceil, q_3=\lceil 1.5 k\rceil$, $q_4= \lceil 1.75 k\rceil$, and $q_5 =2k$).}
\label{fig-parameter1}
\end{figure}

The third one evaluates the sensitivity of CNHTP and CNOTP to the change of stepsize $\lambda$, which varies as $\lambda=1,2, \ldots, 7$ while $q=k$, $\gamma=0.01$ and $\alpha=1$.
As shown in Fig. \ref{fig-parameter2}, CNOTP exhibits relative insensitivity to changes in $\lambda$, demonstrating its robustness across a wide range of stepsizes compared to CNHTP.
This figure also shows that larger values of $\lambda$, such as $\lambda=4$, may lead to improved recovery performance, particularly for CNHTP, although it does not fall within the theoretical range for $\lambda$ in our main results established in Section \ref{theo}.
This indicates that the condition in our main result might have room to improve, and we can leave this as the future work.
Motivated by this observation, we adopt $q=k$ and $\lambda=4$ in the subsequent experiments.

\begin{figure}
\centering
\subfigure[CNHTP]{
\includegraphics[width=0.48\textwidth]{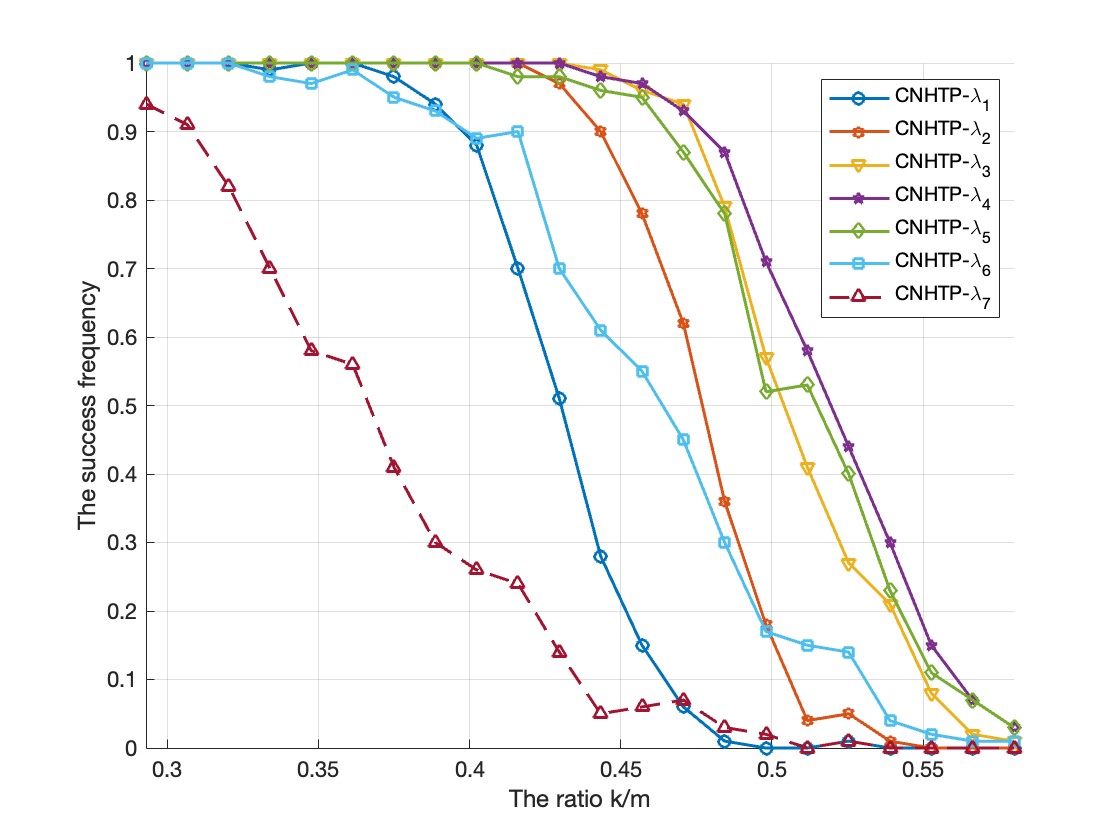}}
\subfigure[CNOTP]{
\includegraphics[width=0.48\textwidth]{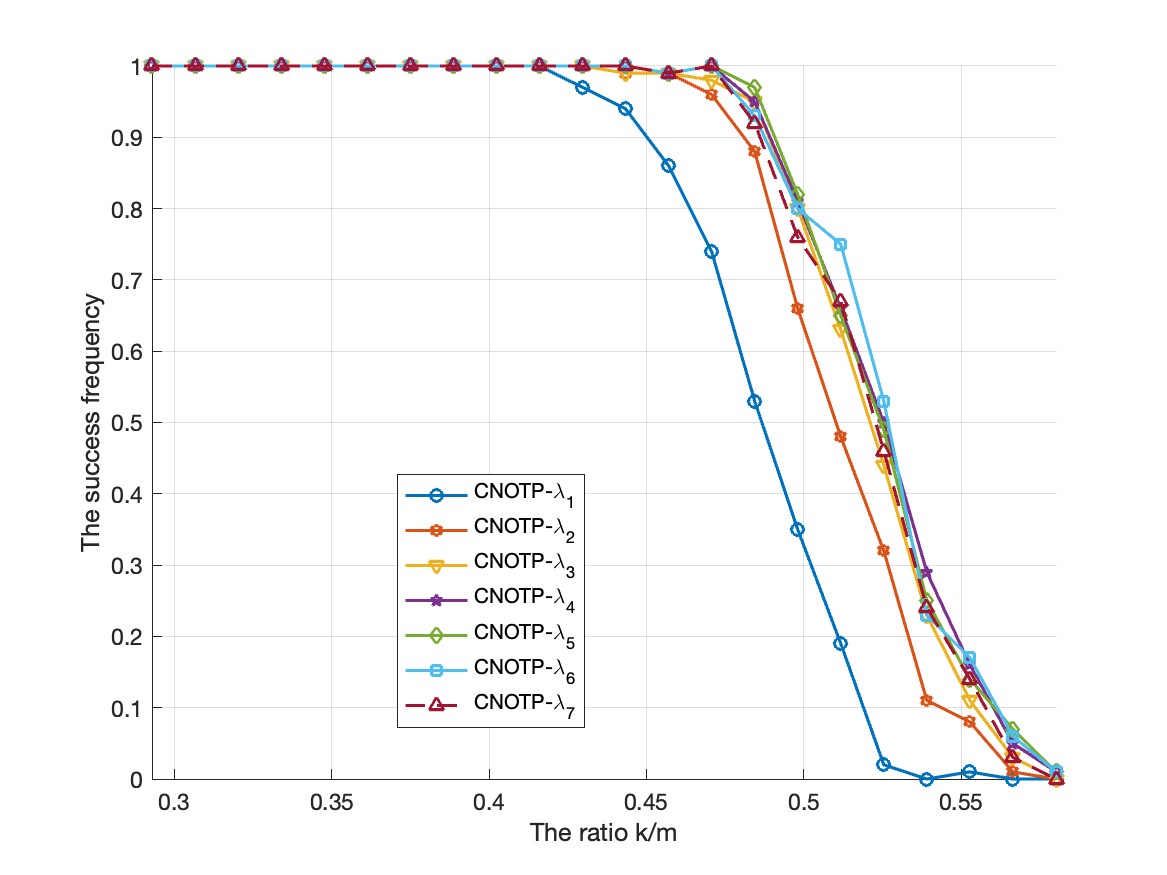}}
\caption{Performance comparison of CNHTP and CNOTP with parameters $q=k$, $\alpha=1$, $\gamma=0.01$ and different values of $\lambda$ ($\lambda=1,2,\dots, 7$).}
\label{fig-parameter2}
\end{figure}

\subsection{Phase transition curves} \label{phase-transition}
In the second experiment, we compare CNHTP and CNOTP with several established sparse optimization algorithms across varying sampling ratios for the matrix and sparsity levels for the target vector.
The vector dimension is set to be $n=1024$.
The horizontal axis in Fig. \ref{figure-phase} corresponds to the sampling ratio $\delta=m / n$, where $m$ is the number of rows of matrix $A$.
We sample a total of 20 values of $\delta$: 10 uniformly in $[0.1,0.3]$ and 10 in $(0.3,0.6]$.
Lower values of $\delta$ (towards the left) correspond to fewer observations.
The vertical axis represents the sparsity ratio $\rho=k / m$, uniformly discretized into 40 values, with smaller $\rho$ indicating sparser vectors.
Two types of random matrices are used: Gaussian matrices with entries i.i.d. drawn from $\mathcal{N}(0,1 / m)$, and Bernoulli matrices with entries taking values $\pm 1 / \sqrt{m}$ with equal probability, i.e., $\mathbb{P}\left(A_{i j}= \pm 1 / \sqrt{m}\right)=1 / 2$.
Observations are generated as $y=A x+10^{-4} \eta$, where $\eta \sim \mathcal{N}(0, I)$ is a Gaussian noise vector.
CNHTP and CNOTP are evaluated using $(q, \lambda, \alpha, \gamma)=(k, 4,1,0.01)$ and compared with four algorithms: PGROTP \cite{meng2022partial} with $(q, \lambda)=(k, 2)$, NTROTP \cite{meng2022newton} with $(\lambda, \epsilon)=\left(1,10^{-3}\right)$, ECDOMP with $\gamma=0.7$ \cite{zhao2023dynamic}, and SP \cite{dai2009subspace}, where $k$ denotes the vector sparsity level.
Each algorithm is executed up to 30 iterations or until the stopping criterion \eqref{stop-criteria} is met.
For each ($\delta, \rho$) pair, we generate 10 random problem instances and record the number of successful recoveries.
Then, for each $\delta$, a logistic regression model is fitted to the success frequencies as a function of $\rho$.

Fig. \ref{figure-phase} demonstrates, for each $\delta$, the estimated value of $\rho$ at which the success frequency reaches 0.9.
This phase transition figure shows that a higher recovery curve indicates stronger performance, as it indicates that the algorithm is able to solve a broader range of problem instances.
The overall performance of CNHTP and CNOTP is comparable to those existing methods tested in this experiment.

\begin{figure}
\centering
\subfigure[Gaussian matrices]{
\includegraphics[width=0.48\textwidth]{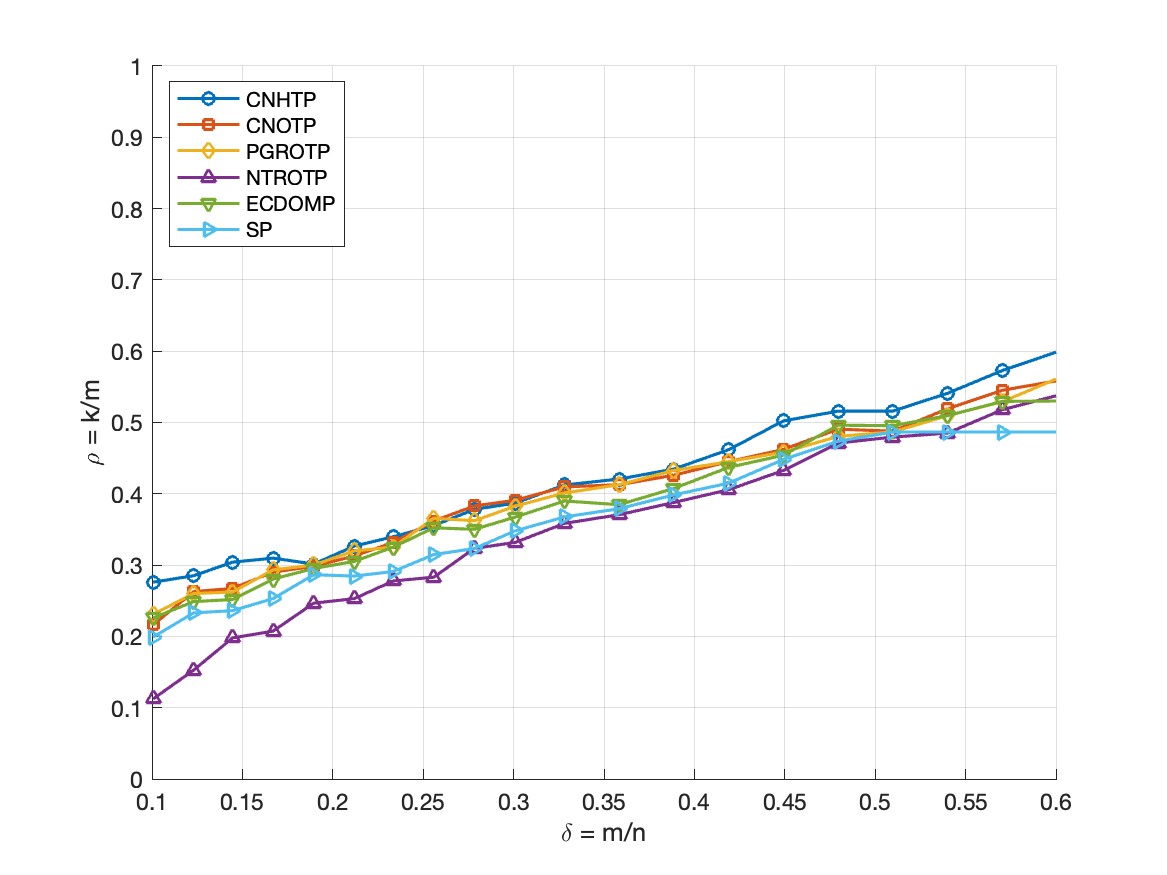}}
\subfigure[Bernoulli matrices]{
\includegraphics[width=0.48\textwidth]{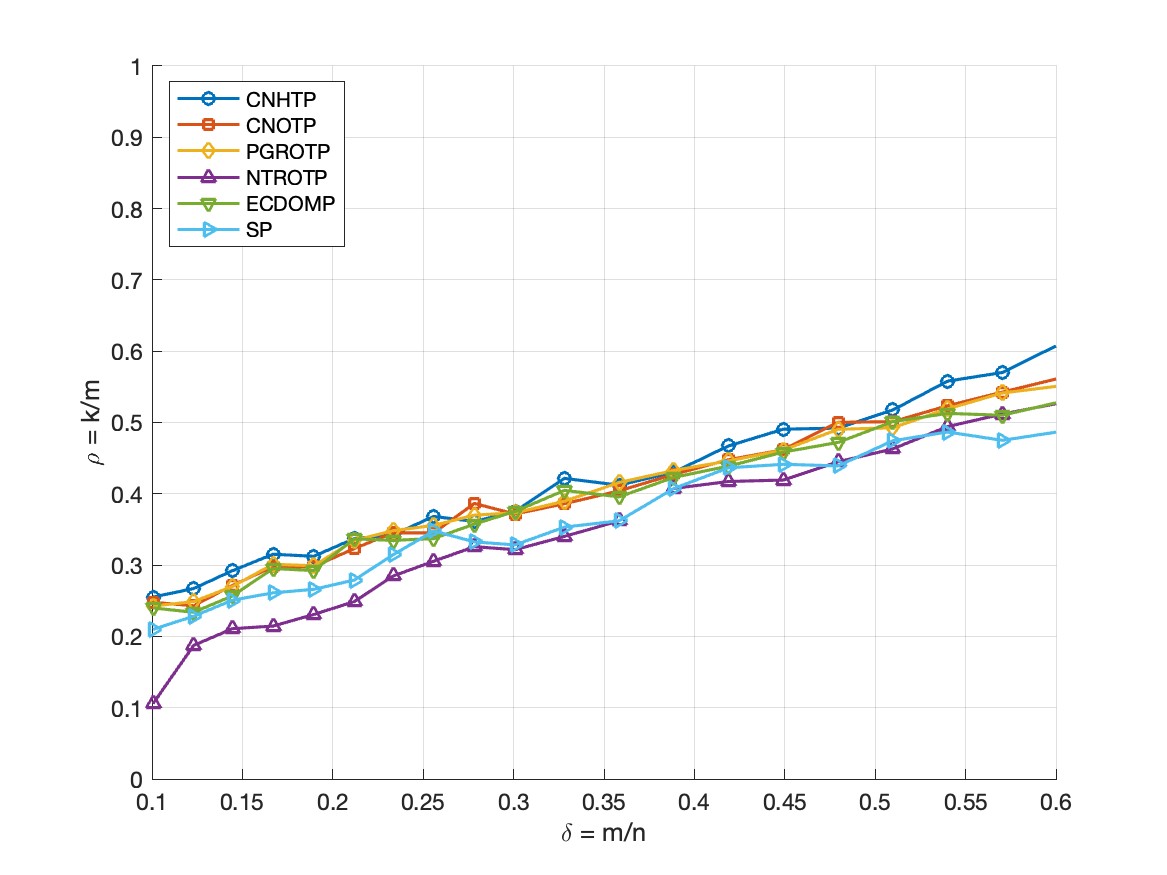}}
\caption{Phase transition curves for sparse signal recovery with success frequency of 0.9.} \label{figure-phase}
\end{figure}

It should be pointed out that this work only analyzes the theoretical performance of the proposed algorithms when the parameter $\gamma$ is taken to be a positive number. The convergence of our algorithms in the extreme case $\gamma = 0$ has not been established, since the key Lemma 3.3 established in this work only applies to the case $\gamma > 0$.
However, it is worth noting that CNOTP with $\gamma = 0$ has relatively low computational cost compared to the case $\gamma \neq 0$, as it significantly reduces the cost of solving the first subproblem in Step 3 of Algorithm \ref{CNOTP}. This advantage can be seen from Fig. \ref{figure-timeratio}, which illustrates the ratio of the average computation time of NTROTP to that of CNOTP for successful signal recovery by using Gaussian and Bernoulli matrices, respectively.
It can be seen from Fig. \ref{figure-timeratio} that CNOTP requires less computational time than NTROTP, and that the choice between Gaussian or Bernoulli matrices does not make any noticeable difference to their performance.

\begin{figure}
\centering
\subfigure[Gaussian matrices]{
\includegraphics[width=0.48\textwidth]{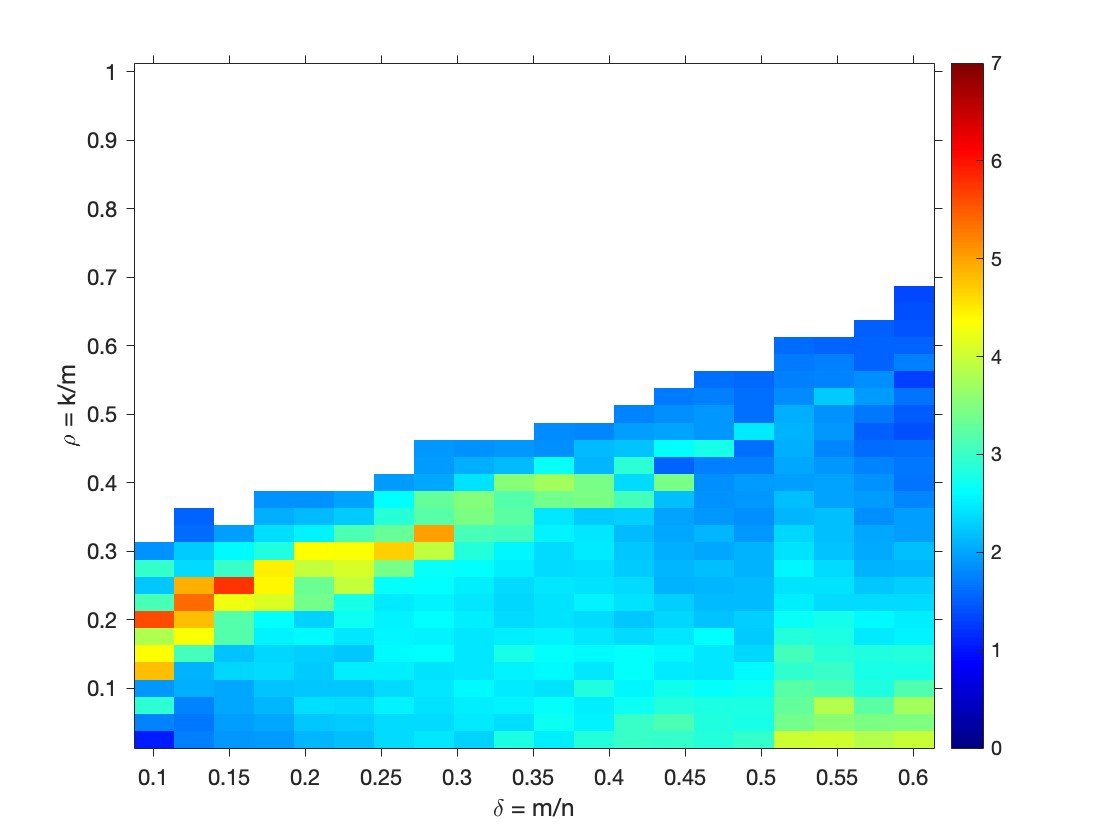}}
\subfigure[Bernoulli matrices]{
\includegraphics[width=0.48\textwidth]{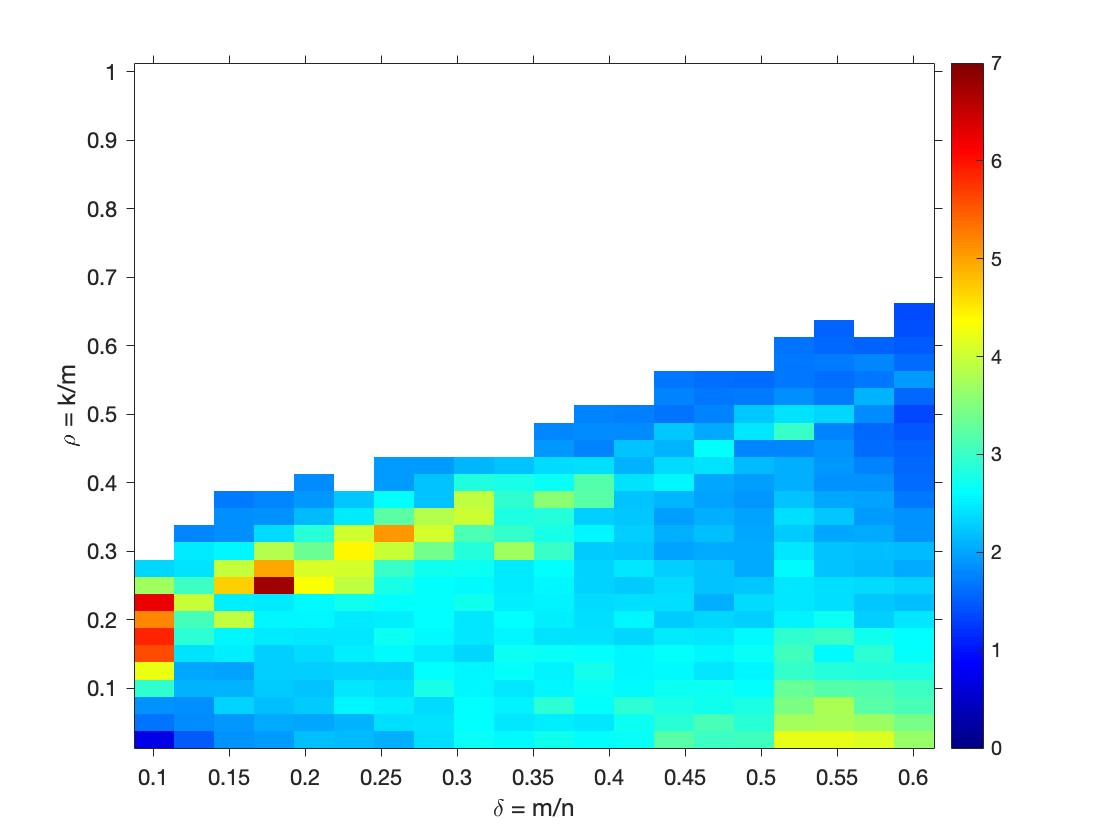}}
\caption{Time consumption ratio NTROTP/CNOTP for successful recovery.} \label{figure-timeratio}
\end{figure}

\section{Summary}
In this work, we have introduced a class of compressed Newton-direction thresholding algorithms for solving sparse optimization problems.
Under the restricted isometry property, we have shown that with suitably chosen parameters the proposed algorithms can be globally converge to the sparse solution of the problem.
Empirical results demonstrate that our methods perform comparably to several existing algorithms in solving the sparse optimization problem typically arising from signal recovery.

\bmhead{Data availability statement}

The data that support the findings of this study are available on request from the authors.


\begin{thebibliography}{999} \label{ref}

\bibitem{beck2009fast} Beck, A., Teboulle, M.: A fast iterative shrinkage-thresholding algorithm for linear inverse problems. SIAM J. Imaging Sci. 2(1), 183-202 (2009)

\bibitem{blumensath2008iterative} Blumensath, T., Davies, M. E.: Iterative thresholding for sparse approximations. J. Fourier Anal. Appl. 14, 629-654 (2008)

\bibitem{blumensath2009iterative} Blumensath, T., Davies, M.E.: Iterative hard thresholding for compressed sensing. Appl. Comput. Harmon. Anal. 27(3), 265-274 (2009).

\bibitem{blumensath2010normalized} Blumensath, T., Davies, M. E.: Normalized iterative hard thresholding: Guaranteed stability and performance. IEEE J. Sel. Top. Signal Process. 4(2), 298-309 (2010)

\bibitem{bouchot2016hard} Bouchot, J. L., Foucart, S., Hitczenko, P.: Hard thresholding pursuit algorithms: Number of iterations. Appl. Comput. Harmon. Anal. 41(2), 412-435 (2016)

\bibitem{candes2005decoding} Candes, E. J., Tao, T.: Decoding by linear programming. IEEE Trans. Inf. Theory 51(12), 4203-4215 (2005)

\bibitem{chen2001atomic} Chen, S. S., Donoho, D. L., Saunders, M. A.: Atomic decomposition by basis pursuit. SIAM Rev. 43(1), 129-159 (2001)

\bibitem{chen2017fast} Chen, J., Gu, Q.: Fast Newton hard thresholding pursuit for sparsity constrained nonconvex optimisation. In: Proc. 23rd ACM SIGKDD Int. Conf. Knowl. Discov. Data Min., pp. 757-766 (2017).

\bibitem{dai2009subspace} Dai, W., Milenkovic, O.: Subspace pursuit for compressive sensing signal reconstruction. IEEE Trans. Inf. Theory 55(5), 2230-2249 (2009)

\bibitem{dembo1982truncated} Dembo, R. S., Steihaug, T.: Truncated-Newton algorithms for large-scale unconstrained optimization. Math. Program. 26(2), 190-212 (1983)

\bibitem{dennies1977quasi} Dennis, Jr, J. E., Moré, J. J.: Quasi-Newton methods, motivation and theory. SIAM Rev. 19(1), 46-89 (1977)

\bibitem{donoho1995denoising} Donoho, D. L.: De-noising by soft-thresholding. IEEE Trans. Inf. Theory 41(3), 613-627 (1995)

\bibitem{donoho2001uncertainty} Donoho, D. L., Huo, X.: Uncertainty principles and ideal atomic decomposition. IEEE Trans. Inf. Theory 47(7), 2845-2862 (2001)

\bibitem{donoho2006stability} Donoho, D. L., Elad, M.: On the stability of the basis pursuit in the presence of noise. Signal Process. 86(3), 511-532 (2006)

\bibitem{donoho2012sparse} Donoho, D. L., Tsaig, Y., Drori, I., Starck, J. L.: Sparse solution of underdetermined systems of linear equations by stagewise orthogonal matching pursuit. IEEE Trans. Inf. Theory 58(2), 1094-1121 (2012)

\bibitem{elad2010sparse} Elad, M.: Sparse and redundant representations: from theory to applications in signal and image processing. Springer (2010)

\bibitem{eldar2012compressed} Eldar, Y. C., Kutyniok, G. (eds.): Compressed Sensing: Theory and Applications. Cambridge University Press (2012)

\bibitem{foucart2011hard} Foucart, S.: Hard thresholding pursuit: an algorithm for compressive sensing. SIAM J. Numer. Anal. 49(6), 2543-2563 (2011)

\bibitem{foucart2013mathematical} Foucart, S., Rauhut, H.: A Mathematical Introduction to Compressive Sensing. Birkhäuser, Basel (2013)

\bibitem{grant2020cvx} Grant, M., Boyd, S.: CVX: Matlab software for disciplined convex programming, version 2.2. http://cvxr.com/cvx (2020).

\bibitem{gui2016feature} Gui, J., Sun, Z., Ji, S., Tao, D., Tan, T.: Feature selection based on structured sparsity: A comprehensive study. IEEE Trans. Neural Netw. Learn. Syst. 28(7), 1490-1507 (2016)

\bibitem{ji2008bayesian} Ji, S., Xue, Y., Carin, L.: Bayesian compressive sensing. IEEE Trans. Signal Process. 56(6), 2346-2351 (2008)

\bibitem{jing2014quasi} Jing, M., Zhou, X., Qi, C.: Quasi-Newton iterative projection algorithm for sparse recovery. Neurocomputing 144, 169-173 (2014)

\bibitem{liao2024subspace} Liao, S., Han, C., Guo, T., Li, B.: Subspace Newton method for sparse group $\ell_0$ optimisation problem. J. Glob. Optim. 90(1), 93-125 (2024)

\bibitem{martens2010deep} Martens, J.: Deep learning via Hessian-free optimization. In: Proc. ICML, vol. 27, pp. 735-742 (2010)

\bibitem{meng2020newton} Meng, N., Zhao, Y. B.: Newton-step-based hard thresholding algorithms for sparse signal recovery. IEEE Trans. Signal Process. 68, 6594-6606 (2020)

\bibitem{meng2022newton} Meng, N., Zhao, Y. B.: Newton-type optimal thresholding algorithms for sparse optimization problems. J. Oper. Res. Soc. China 10(3), 447-469 (2022)

\bibitem{meng2022partial} Meng, N., Zhao, Y. B., Kočvara, M., Sun, Z.: Partial gradient optimal thresholding algorithms for a class of sparse optimization problems. J. Glob. Optim. 84(2), 393-413 (2022)

\bibitem{needell2009cosamp} Needell, D., Tropp, J. A.: CoSaMP: Iterative signal recovery from incomplete and inaccurate samples. Appl. Comput. Harmon. Anal. 26(3), 301-321 (2009)

\bibitem{nocedal2006numerical} Nocedal, J., Wright, S. J.: Numerical Optimization. Springer, New York, NY (2006)

\bibitem{parikh2014proximal} Parikh, N., Boyd, S.: Proximal algorithms. Found. Trends Optim. 1(3), 127-239 (2014)

\bibitem{pati1993orthogonal} Pati, Y. C., Rezaiifar, R., Krishnaprasad, P. S.: Orthogonal matching pursuit: Recursive function approximation with applications to wavelet decomposition. In: Proc. 27th Asilomar Conf. Signals, Systems and Computers, pp. 40-44. IEEE (1993)

\bibitem{sra2011optimization} Sra, S., Nowozin, S., Wright, S.J. (eds.): Optimisation for machine learning. MIT Press (2011)

\bibitem{sun2023heavy} Sun, Z. F., Zhou, J. C., Zhao, Y. B.: Heavy-ball-based optimal thresholding algorithms for sparse linear inverse problems. J. Sci. Comput. 96(3), 93 (2023)

\bibitem{tanner2013normalized} Tanner, J., Wei, K.: Normalised iterative hard thresholding for matrix completion. SIAM J. Sci. Comput. 35(5), S104-S125 (2013)

\bibitem{tibshirani1996regression} Tibshirani, R.: Regression shrinkage and selection via the lasso. J. R. Stat. Soc. Series B Methodol. 58(1), 267-288 (1996)

\bibitem{tropp2004greed} Tropp, J. A.: Greed is good: Algorithmic results for sparse approximation. IEEE Trans. Inf. Theory 50(10), 2231-2242 (2004)

\bibitem{wang2015recovery} Wang, J., Kwon, S., Li, P., Shim, B.: Recovery of sparse signals via generalised orthogonal matching pursuit: A new analysis. IEEE Trans. Signal Process. 64(4), 1076-1089 (2015)

\bibitem{wang2018new} Wang, Q., Qu, G.: A new greedy algorithm for sparse recovery. Neurocomputing 275, 137-143 (2018)

\bibitem{wu2008coordinate} Wu, T. T., Lange, K.: Coordinate descent algorithms for lasso penalised regression. Ann. Appl. Stat. 2(1), 224-244 (2008)

\bibitem{ye2025subspace} Ye, Y., Li, Q.: Subspace Newton's method for $\ell_0$-regularised optimisation problems with box constraint. arXiv:2505.17382 (2025)

\bibitem{yuan2014newton} Yuan, X. T., Liu, Q.: Newton greedy pursuit: A quadratic approximation method for sparsity-constrained optimisation. In: Proc. IEEE Conf. Comput. Vis. Pattern Recognit., pp. 4122-4129 (2014)

\bibitem{yuan2017newton} Yuan, X. T., Liu, Q.: Newton-type greedy selection methods for $\ell_0$-constrained minimisation. IEEE Trans. Pattern Anal. Mach. Intell. 39(12), 2437-2450 (2017)

\bibitem{zhao2018sparse} Zhao, Y. B.: Sparse Optimisation Theory and Methods. CRC Press, Boca Raton, FL (2018)

\bibitem{zhao2020optimal} Zhao, Y. B.: Optimal $k$-thresholding algorithms for sparse optimisation problems. SIAM J. Optim. 30(1), 31-55 (2020)

\bibitem{zhao2021analysis} Zhao, Y. B., Luo, Z. Q.: Analysis of optimal thresholding algorithms for compressed sensing. Signal Process. 187, 108148 (2021)

\bibitem{zhao2023dynamic} Zhao, Y. B., Luo, Z. Q.: Dynamic orthogonal matching pursuit for sparse data reconstruction. IEEE Open J. Signal Process. 4, 242-256 (2023)

\bibitem{zhao2023improved} Zhao, Y., Luo, Z.: Improved RIP-based bounds for guaranteed performance of two compressed sensing algorithms. Sci. China Math. 66(5), 1123-1140 (2023)

\bibitem{zhou2021newton} Zhou, S., Pan, L., Xiu, N.: Newton method for $\ell_0$-regularized optimisation. Numer. Algorithms 88(4), 1541–1570 (2021)

\bibitem{zhou2021global} Zhou, S., Xiu, N., Qi, H. D.: Global and quadratic convergence of Newton hard-thresholding pursuit. J. Mach. Learn. Res. 22(12), 1-45 (2021)

\bibitem{zhou2022gradient} Zhou, S.: Gradient projection Newton pursuit for sparsity constrained optimisation. Appl. Comput. Harmon. Anal. 61, 75-100 (2022)

\end{thebibliography}
\end{document}